\documentclass{article}
\usepackage[utf8]{inputenc}

\usepackage[english]{babel}
\usepackage{graphicx}
\usepackage{subcaption}
\usepackage{algorithm}
\usepackage{algorithmic}
\usepackage[a4paper,top=3cm,bottom=2cm,left=3cm,right=3cm,marginparwidth=1.75cm]{geometry}

\usepackage{amsmath}
\usepackage{ amssymb }
\usepackage{graphicx}
\usepackage[colorlinks=true, allcolors=blue]{hyperref}
\usepackage{amsfonts}
\usepackage{amsthm}
\usepackage{hyperref}
\usepackage{natbib}

\usepackage{tikz}
\usetikzlibrary{shapes.geometric, arrows.meta, positioning}

\newtheorem{theorem}{Theorem}[section]
\newtheorem{lemma}{Lemma}[section]
\newtheorem{corollary}{Corollary}[section]
\newtheorem{definition}{Definition}[section]

\newcommand{\cX}{{\cal X}}

\newcommand{\cC}{{\cal C}}
\newcommand{\cA}{{\cal A}}

\newcommand{\cF}{{\cal F}}

\newcommand{\cV}{{\cal V}}

\newcommand\F{\mathcal{F}}
\renewcommand\S{\mathcal{S}}

\renewcommand\a{\mathbf{a}}
\renewcommand\c{\mathbf{c}}

\newcommand\y{\mathbf{y}}
\newcommand\x{\mathbf{x}}
\renewcommand\v{\mathbf{v}}
\renewcommand\u{\mathbf{u}}
\newcommand\f{\mathbf{f}}
\newcommand\z{\mathbf{z}}

\newcommand\A{\mathbf{A}}
\newcommand\U{\mathbf{U}}

\newcommand\X{\mathbf{X}}

\newcommand\I{\mathbb{I}}
\newcommand\E{\mathbb{E}}
\renewcommand\P{\mathbb{P}}
\newcommand\R{\mathbb{R}}

\title{Towards Automated Confidence Bound Provers and Searchers}
\author{My Phan \and Erik Learned-Miller\thanks{University of Massachusetts Amherst }}

\begin{document}

\maketitle
\begin{abstract}
In this work we lay the groundwork for automating the process of finding and proving the validity of lower confidence bounds of the mean. Our key finding is based on the observation that finding an optimal confidence bound under certain conditions can be formulated as an optimization problem. We use this observation to show that any valid confidence bound (such as Hoeffding's) must be a relaxation of a certain optimization problem. To automatically find and prove confidence bounds, we need to automate the process of defining and finding such a relaxation. We define a family of relaxations parameterized by a function called the order function. This family of relaxations can approximate any other target relaxation such as Hoeffding's by using the target as the order function. When the order function is linear (such as the sample mean in the case of Hoeffding's), our relaxation is a linear-size mixed-integer linear program.
\end{abstract}
\section{Introduction} 
\subsection{Goals}

In this paper we study valid lower confidence bounds (Definition~\ref{def:bound}) for parameters of the distribution (such as the mean or the quantile) from $n$ independently and identically distributed (i.i.d.) samples with replacement.

Let $n$ be the sample size. We use bold-faced letters to denote a vector of size $n$ and normal letters to denote a scalar. Uppercase letters denote random variables and lowercase letters denote values taken by them. For example, $X_i \in \R$ and $\X = (X_1, ..., X_n) \in \R^{n}$ are random variables. $x_i \in \R$ is a value of $X_i$, and $\x = (x_1, ..., x_n) \in \R^{n}$ is a value of $\X$.
\begin{definition}[Valid $(1-\alpha)$ lower confidence bound (LCB)]
\label{def:bound}
Let $b: \R^n \rightarrow \R$ be a function of the sample. Let $\mathcal{F}$ be a set of distributions, $\theta: \mathcal{F} \rightarrow \R$ be a parameter of the distribution and $\alpha \in (0,1)$ be the significance level. Let $\X$ denote $n$ i.i.d. samples from $F$. 

A function $b$ is called {\em a valid lower confidence bound of parameter $\theta$ on a set of probability distributions $\mathcal{F}$} if for all distributions $F\in \mathcal{F}$,  
\begin{equation}
\label{eq:valid}
\P_F( b(\X) \le \theta(F)) \geq 1-\alpha,
\end{equation}
where $\P_F$ denotes the probability under distribution $F$. 
\end{definition}
For example, the Hoeffding's bound $h(\x) = \frac{1}{n} (\sum_{i=1}^n x_i) - \sqrt{\frac{\log(1/\alpha)}{2n}}$ is a valid lower confidence bound for the set $\mathcal{F}$ of all distributions on $[0,1]$ because $\P_F(h(\X) \le \mu(F)) \ge 1 - \alpha$. For brevity we use the terms valid lower confidence bound, valid confidence bound and valid bound interchangeably. 

We say that a function $f$ {\em is ordered by $T$} (or \emph{$T$-order}) if, for any two samples $\x$ and $\y$,  $T(\x) \le T(\y)$ implies $f(\x) \le f(\y)$. For example, Hoeffding's bound \citep{Hoeffding1963} is a function ordered by the sample mean. We call $T$ an \emph{order function}. Let $f(\x)$ be any valid lower confidence bound ordered by $T$ on a set of distributions $\F$. We show that we can define a valid lower confidence bound $b^{\F}_T(\x)$ from $\F$ and $T$ such that $b^{\F}_T(\x)$ is ordered by $T$ and $\forall \x \in \R^n, ~ b^{\F}_T(\x) \ge f(\x)$ for any $f$ ordered by $T$ (Theorem~\ref{thm:optimal}). Therefore $b^{\F}_T(\x)$ is optimal among all valid lower confidence bounds ordered by $T$, called the \textbf{$T$-order optimal bound} or \textbf{$T$-optimal bound}. The analysis of bounds ordered by a given function goes back at least to~\cite{Fienberg77}. More recently, \cite{Learned-Miller2025} did an extensive analysis of multinomial bounds (over known support) conditioned on orderings.

\emph{The key observation is that we could compute the $T$-order optimal bound for any sample $\x$ by solving an optimization problem (Definition~\ref{def:optimal_bound}). Therefore any valid bound such as Hoeffding's is a relaxation of this optimization problem (manually derived and proven by the scientist). By automating the relaxation, we can automate the process of finding and proving valid bounds.}

The above result can potentially be applied to answer the following questions: 
\begin{enumerate}

    \item \textbf{Automated Bound Tuning.}  Given a valid lower confidence bound $f(\x)$, does there exist valid lower confidence bound $b(\x)$ such that $b(\x) \ge f(\x) ~\forall \x \in \R^n$? Can we compute it?

Surprisingly, in this work, we find the answer to be positive: If $f(\x)$ is a valid lower confidence bound, then $\forall \x~\in \R^n, f(\x) \le b^{\F}_f(\x)$ and therefore $b^{\F}_f(\x)$ is a better {or equal} valid lower confidence bound. We show a way to compute a lower bound of $b_f$ when $\F$ is the set of all distributions on $[0,1]$. 
    
        \item \textbf{Automated Bound Conversion.}  Given a function $f(\x)$ that is not a valid bound, does there exist a valid bound $b(\x)$ that is ordered by $f$? Can we compute it? 

In this work, we find the answer to be positive: If $f(\x)$ is not a valid lower confidence bound, then  $b^{\F}_f(\x)$ is a valid lower confidence bound ordered by $f$. We show a way to compute a lower bound of $b_f$ when $\F$ is the set of all distributions on $[0,1]$.

    \item \textbf{Automated Bound Proving.} Given a function $f: \R^n \rightarrow \R$, is it a valid lower confidence bound of parameter $\theta$ on distribution set $\F$?

    \begin{itemize}
\item Proving that $f$ is a valid lower confidence bound on $\F$ is equal to proving that $\forall \x, ~f(\x) \le b^{\F}_f(\x)$. For examples, the proofs of Hoeffding's valid lower confidence bound $h(\x)$ is equal to the proof that $\forall \x, ~h(\x) \le b^{\F}_h(\x)$ where $\F$ is the set of all distributions on $[0,1]$.

\item Proving  that $f$ is not a valid lower confidence bound is equal to proving  that $\exists \x, f(\x) > b^{\F}_T(\x)$. 
\end{itemize}
\item \textbf{Automated Bound Search.}  We can envision an agent that automatically searches for valid lower confidence bounds according to an objective. 

 This problem is an interesting future application of our result. In this work, we do not focus on this problem. The agent search for the bound from a function class (such as polynomials or neural networks) to satisfy the objective 
where $L$ is a loss function: 
\begin{align}
\min_f L(f) \text{ such that $f$ is a valid bound}. 
\end{align}

Using our method, we can covert the problem
to
\begin{align}
\min_T L(b_T)
\end{align}
because $b_T$ is guaranteed to be a valid bound. 

Existing works in statistics define many ways to evaluate confidence bounds, such as Uniformly Most Powerful test and Uniformly Most Accurate interval \citep{CaseBerg:01}, or minimum deterministic width \citep{shekhar2023near}. Existing papers also evaluate confidence bounds' performance by doing experiments on a set of distributions \citep{waudbysmith2021estimating}. These objectives can be formalized as loss functions and constraints. For examples, the loss functions can be the sum of the expected performances on uniform distribution and beta distribution.

\end{enumerate}

\subsection{Results}
$b^{\F}_T(\x)$ is the result of an optimization problem which could be {computationally} difficult to solve. When  $\F$ is the set of all distributions on $[0,1]$ and $T$ is monotonic, we present a lower confidence bound $\ell^{\text{MILP}}_T(\x)$ of $b^{\F}_T(\x)$, computed from a linear-sized Mixed-Integer Linear Program. For the applications above:

\begin{enumerate}
\item \textbf{Automated Bound Tuning.} Given a sample $\x$ and a valid bound $T(\x)$, we use $T(\x)$ as the ordering function and derive a valid bound $\ell^{\text{MILP}}_T(\x)$, which is a lower bound of the $T$-optimal bound $b_T(\x)$.  Therefore it is sometimes better than $T(\x)$. We leave the derivation of $\ell^{\text{MILP}}_T(\x)$ that can be arbitrarily close to $b_T(\x)$ and therefore always equal to or better than $T(\x)$ to future works. 

\item \textbf{Automated Bound Conversion.} Given a sample $\x$ and a function $T(\x)$ (not necessarily a valid bound), we can derive a valid bound $\ell^{\text{MILP}}_T(\x)$ ordered by $T$. 

\item \textbf{Automated Bound Proving.} To prove that $T$ is a valid lower confidence bound on $\F$, we can prove that $\forall \x, ~T(\x) \le \ell^{\text{MILP}}_T(\x)$.

In this paper we show how to compute $\ell^{\text{MILP}}_T( \x)$ for a single sample $\x$. To prove that $\forall \x, ~T(\x) \le \ell^{\text{MILP}}_T( \x)$, we need to make additional assumptions. For examples, to show that $\forall \x,  \ell^{\text{MILP}}_T(\x) - T(\x) \ge 0$ using the oracle that computes $\Delta(\x):= \ell^{\text{MILP}}_T( \x) - T(\x)$ for a given $\x$, we need to make assumptions about $\Delta(\x)$ such as Lipschitz continuity. We leave this for future works.
\end{enumerate}

\textbf{Additional Potential Benefit.} In addition to the listed potential applications, there is another advantages of our method: If $\cF$ is the set of all distributions on a support $\cX$ that is more fine-grained then $[0,1]$ (such as $[0,0.2] \cup [0.5, 0.6]$) our bound $\ell^{\text{MILP}}_T( \x)$ computed from $\cX$ can potentially make use of this information to get tighter bound. Bounds such as Hoeffding's can not take advantage of this information.

\textbf{Drawbacks. } There are also several drawbacks of our method: 
\begin{itemize}
    \item $\ell^{\text{MILP}}_T(\x)$ can be expensive to computed. However the bound for each value of $T(\x)$ needs to be computed only once and can be saved in a table indexed by $T(\x)$. An user with a new sample $\x$ only needs to compute $T(\x)$ and use the table to get the value of $\ell^{\text{MILP}}_T(\x)$. 
    \item The bound $\ell^{\text{MILP}}_T(\x)$ can only be computed for the set of distributions on bounded support $[0,1]$. We leave the case of support $[0, \infty)$ to future works. 
    \item Our bounds do not have a closed-form formula so they cannot be directly used in mathematical proofs. However they can potentially be used to prove other bounds that have closed-form formulae. 
    \item We use samples of uniform i.i.d. random variables to compute  $\ell^{\text{MILP}}_T(\x)$. We need to account for the extra randomness from the samples of uniform random variables (Section~\ref{sec:analysis}). 
\end{itemize} 
The structure of the paper is as follows. In Section~\ref{sec:def} we define the optimal bound $b^{\F}_T(\x)$ and shows that it is optimal. In Section~\ref{sec:MILP} we convert $b^{\F}_T(\x)$ to chance-constrained programs and then show how to compute its lower bound using Mixed-Integer Linear Programs. 
In Section~\ref{sec:exp} we perform experiments to compute the bound when the order function $T$ is the sample mean, Anderson's bound (a proven valid bound) and Gaffke's bound (an unproven bound with good performance). 

\section{ A \texorpdfstring{$T$}{T}-Optimal Confidence Bound}
\label{sec:def}

We adopt the notations and build some of our results from \citep{phan2021practical}. 
 Let $X_{(1)}\leq X_{(2)}\leq ...\leq X_{(n)}$ denote the order statistics of a random variable and  $x_{(1)}\leq x_{(2)}\leq ...\leq x_{(n)}$ the order statistics of a specific sample. For any $2$ vectors $\x, \y \in \R^n$, we write $\x \leq \y$ or $\x \preceq \y$ if $x_{(i)} \leq y_{(i)}~\forall i, 1 \le i \le n$.

We consider sampling with replacement. Given a sample $\X = \x$ of size $n$ and a confidence level $1-\alpha$, we would like to calculate a valid lower confidence bound for a parameter $\theta$ of the distribution. Let $F$ be the cumulative distribution function (CDF) of $X_i$, i.e., the true distribution. Let $\mathcal{F}$ be the set of distributions that we know $F$ belongs to. Given a sample $\x$, $\alpha$, $\mathcal{F}$ and  any function $T: \R^n \rightarrow \R$ we will calculate an valid lower confidence bound $b^{\F}_T(\x, \alpha) $ for $\theta(F)$. 
Throughout the paper we will omit the superscripts and subscripts when the context is clear.

In Section~\ref{sec:sub_def} we define the bound (Definition~\ref{def:optimal_bound}) and show that it is valid. In Section~\ref{sec:opt} we show that it is optimal among all $T$-order bounds (Theorem~\ref{thm:optimal}). 

In Section~\ref{sec:sample_mean_order} we calculate the sample-mean-optimal bound, which has a closed-form formula. This bound is always better than Hoeffding's since Hoeffding's is ordered by the sample mean. 

\subsection{Definition}
\label{sec:sub_def}
The following definition gives a valid lower confidence bound of a parameter $\theta$ (such as the mean or the quantile) of the distribution for a sample $\x$. The definition is a minor modification of pivoting the CDF and inverting a hypothesis test (Appendix~\ref{app:generalization}). Instead of constructing a confidence set of the parameter $\theta$, we construct a confidence set of the CDF and then find the lower bound of the parameter $\theta$ in the confidence set of the CDF. (\cite{Learned-Miller2025} refers to this as the \textit{central optimization problem}.) Let $F_X$ denote the CDF of $X$. Let $\P_F$ denote the probability with respect to CDF $F$ and $\P_X$ denote the probability with respect to random variable $X$. Given a CDF $F_X$ of random variable $X$, let $F_X^{\ge}(x):= \P_F(X \ge x)$.  Let $F_{T(\X)}$ denote the CDF of $T(\X)$, and therefore $F^{\ge}_{T(\X)}(y) := \P_{F_{T(\X)}}(T(\X) \ge y)$. Let $\alpha \in (0,1)$ be the significance level ($1-\alpha$ is the confidence level).

 \begin{definition}
 \label{def:optimal_bound}
 Let $\F$ be a set of distributions. Let $\theta$ be a parameter of distributions and suppose that $\theta(F)$ exists for all $F \in \F$. Let $T: \R^n \rightarrow \R$ be a statistic of a sample of size $n$. 
Given a sample $\x \in \R^n$,  define
$$
b^{\mathcal{F}}_T(\x, \alpha) :=\inf_F \; \theta(F)
$$
subject to 
\begin{enumerate}
\item $F \in \mathcal{F}$
    \item $\P_F(T(\X) \ge T(\x)) > \alpha$
\end{enumerate}
and let $b^{\mathcal{F}}_T(\x, \alpha) := \infty$ if the problem is infeasible.
\end{definition}

We omit the superscripts and subscript and the parameter $\alpha$ when the context is clear. 
We will show in Theorem~\ref{thm:correct} that $b_T(\X)$ is smaller than $\theta(F)$ with probability at least $1-\alpha$.

\begin{theorem}[Validity]
\label{thm:correct}
Let $\F$ be a set of distributions and suppose that $\theta(F)$ exists for all $F \in \F$. Let $T: \R^n \rightarrow \R$ be a statistic of a sample of size $n$. For any $F \in \mathcal{F}$, the bound produced by Definition~\ref{def:optimal_bound} satisfies
\begin{align*}
\P_F(b^{\F}_T(\X, \alpha) \le \theta(F)) \ge 1-\alpha. 
\end{align*}
\end{theorem}

Note that the theorem can be applied to any set $\mathcal{F}$ such as: $\F$ is the set of all distributions on $[0,0.5] \cup [0.6,1]$, $\F$ is the set of all discrete distributions on $\{0, 0.2, 1\}$, $\F$ is the set of all continuous distributions on $[0,1]$, $\F$ is the set of all distributions on $[0, \infty)$ and $\F$ is the set of all Gaussian distributions $\mathcal{N}(\mu, 1)$ with unknown mean $\mu$.

\subsection{\texorpdfstring{$T$}{T}-Optimality}
\label{sec:opt}
First we show that the bound in Definition~\ref{def:optimal_bound} is ordered by $T$. 
\begin{lemma}
\label{lem:helper}
Let $\F$ be a set of distributions and suppose that $\theta(F)$ exists for all $F \in \F$. Let $T: \R^n \rightarrow \R$ be a statistic of a sample of size $n$.
For any $2$ samples $\x, \y \in \R^n$, if $T(\x) \le T(\y)$ then $b^{\F}_T(\x, \alpha) \le b^{\F}_T(\y, \alpha)$.
\end{lemma}

Now we show that for any sample, the bound in Definition~\ref{def:optimal_bound} is the tightest with respect to all bounds ordering the samples by $T$. 

\begin{theorem}[$T$-order Optimality]
\label{thm:optimal}
Let $\F$ be a set of distributions and suppose that $\theta(F)$ exists for all $F \in \F$. Let $T: \R^n \rightarrow \R$ be a statistic of a sample of size $n$. Let $g:\R^n \rightarrow \R$  be a $(1-\alpha)$-valid lower confidence bound on $\F$ that is order by $T$. Then 
\begin{align}
~\forall \x \in \R^n, g(\x) \le b^{\F}_T(\x, \alpha) . 
\end{align}

We call $b^{\F}_T(\x, \alpha) $ the \emph{$T$-order optimal bound} or \emph{$T$-optimal bound}. 
\end{theorem}

If the order function $T$ is a valid lower confidence bound (such as Hoeffding's \citep{Hoeffding1963} when $\mathcal{F}$ is the set of all distributions on $[0,1]$), our bound $b^{\F}_T$ is an equal or better valid lower confidence bound. If $T$ is not a valid lower confidence bound (such as Student's $t$ \citep{student1908probable} when $\mathcal{F}$ is the set of all distributions on $[0,1]$), our bound $b^{\F}_T$ is a valid lower confidence bound. Therefore $b^{\F}_T$ is always an improvement over $T$.

\begin{corollary}
\label{cor:proof_equiv}
Let $\F$ be a set of distributions and suppose that $\theta(F)$ exists for all $F \in \F$. Let $T: \R^n \rightarrow \R$ be a statistic of a sample of size $n$.
$T$ is a $1-\alpha$-valid lower confidence bound on $\F$ if and only if $\forall \x \in \R^n, T(\x) \le b^{\F}_T(\x, \alpha)$. 
\end{corollary}

Any valid confidence bound $T$ is a relaxation of the optimization problem in Def.~\ref{def:optimal_bound}. Therefore we can perform bound tuning: given a valid confidence bound $T$, we can output $b_T$, which is an equal or better valid bound.

\subsection{Sample Mean as the Order Function}
\label{sec:sample_mean_order}
In this section we derive the optimal bound where the parameter $\theta$ is the distribution's mean, the order function $T$ is the sample mean and $\F$ is the set of all distributions on $[0,\infty)$ or $[0,1]$. We omit the subscript $T$ in $b^{\F}_T$ and use $b^{[a,b]}$ to denote the bound in Definition~\ref{def:optimal_bound} when $\F$ is the set of distributions on the interval $[a,b]$. The following bounds are derived from the tail inequality in \citep{luczak2016maximaltailprobabilitysums}: 
\begin{theorem}[Derived from \citep{luczak2016maximaltailprobabilitysums}]
\label{thm:luczak_positive}
Let $T$ be the sample mean. Suppose (i) $n=2$ and $\alpha < 16/25$ or (ii) $n=3$ and $ \alpha < 0.63$ or (iii) $5 \le n \le 10,000 $ and $\alpha \le 0.3$.

Let $\F$ be the set of all distributions on $[0,\infty)$. Then the $T$-optimal bound is 
\begin{align}
    b^{[0, \infty)}(\x) =
    (\sum_{i=1}^n x_i)(1 - (1-\alpha)^{1/n}).
\end{align}

Let $\F$ be the set of all distributions on $[0,1]$ and suppose that $\frac{\sum_{i=1}^n x_i}{n} \le 1/n$. Then the $T$-optimal bound is
\begin{align}
    b^{[0, 1]}(\x) =
    (\sum_{i=1}^n x_i)(1 - (1-\alpha)^{1/n}). 
\end{align}

\end{theorem}
\textbf{Note. }Recall that Hoeffding's is ordered by the sample mean, and therefore is always worse than or equal to this bound when the condition of the theorem is satisfied.

\section{A Lower Bound of the \texorpdfstring{$T$}{T}-Optimal Bound}
\label{sec:MILP}
In the previous section we derive a sample-mean-order bound by using the calculation of the tail $\sup_{F: \mu_F \le v} \P_F(\sum_i X_i \ge c)$ from \citep{luczak2016maximaltailprobabilitysums}.  The Hoeffding's bound is also derived from Hoeffding's inequality. Such method is not general because it requires us to manually calculate and prove the result. We aim to automate this process by transforming the problem of calculating $b_T(x)$ to an optimization problem where $T$ appears in the constraint. Then we can solve the problem for a large class of $T$ including linear functions. 

In this section we derive a lower bound of $b_T(\x)$
that can be computed efficiently as mixed-integer linear programs (MILP) for a large class of order function $T$. We define the following property of $T$ that will be required in several steps. 
\begin{definition}[Monotonic Statistics]
    A statistics $T: \R^n \rightarrow \R$ is called \emph{monotonic} if $T$ satisfies the condition: for any $\y, \z \in \R^n$, if $\y \leq\z$ then $T(\y) \le T(\z)$. 
\end{definition}
 In Section~\ref{sec:discrete_bound} we derive a lower bound and an upper bound of the $T$-optimal bound  by searching through discrete distributions. In Section~\ref{sec:chance_constraint} we convert the lower bound to a chance-constrained program. In Section~\ref{sec:sample_approximation} we show how to compute the lower bound using sample approximation, turning the problem into a MILP. The entire process is in Figure~\ref{fig:process}\footnote{Gemini generated the code template for this figure.}. 


\begin{figure}[htbp]
    \centering 
    
    \begin{tikzpicture}[
        block/.style={
            rectangle, 
            draw=blue!70!black, 
            fill=blue!10, 
            thick, 
            text width=2.2cm, 
            align=center, 
            minimum height=1.2cm,
            rounded corners=2pt
        },
        arrow/.style={
            -{Stealth[scale=1.2]}, 
            thick,
            draw=gray!80!black
        },
        label/.style={
            font=\small\sffamily,
            text=black!80
        }
    ]

        \node (block1) [block] {$T$-optimal bound $b_T({\alpha- \delta})$ (Thm.~\ref{thm:optimal})};
        \node (block2) [block, right=2cm of block1] {Discrete bound $\ell^{\text{Discrete}}_T({ \alpha- \delta})$ (Lemma~\ref{lem:discrete_bound})};
        \node (block3) [block, right=4cm of block2] {Chance-constrained bound $\ell^{\text{Uniform}}_T({ \alpha- \delta})$ (Lemma~\ref{lem:chance_constrained_bound})};
        \node (block4) [block, below=1.5 cm of block3] {Sample-approximation bound $\ell^{\text{Sample}}_T({ \epsilon})$ (Lemma~\ref{lem:sample_approximation_bound})};
        \node (block5) [block, left=4 cm of block4]  {MILP bound $\ell^{\text{MILP}}_T({ \epsilon})$ (Lemma~\ref{lem:milp_bound})};

        \draw [arrow] (block1) -- node[label, above] {discretization} node[label, below] {param: $m$}(block2);
        \draw [arrow] (block2) -- node[label, above] {constraint conversion } node[label, below] {Lemma~\ref{lem:constraint_equal}} (block3);
        \draw [arrow] (block3) -- node[label, left] {sample approximation }  node[label, right] {param: $N$} (block4);
        \draw [arrow] (block4) -- node[label, above] {MILP gap} node[label, below] {param: $\gamma$}(block5);
    \end{tikzpicture}

    \caption{ Converting the $T$-optimal bound to a linear-sized MILP.
    \label{fig:process}. We modify the significance level $\alpha$ to smaller values $\alpha - \delta$ and $\epsilon \le \alpha - \delta$ to ensure that the overall confidence level is still $1 - \alpha$ (Section~\ref{sec:analysis}). }
    
\end{figure}
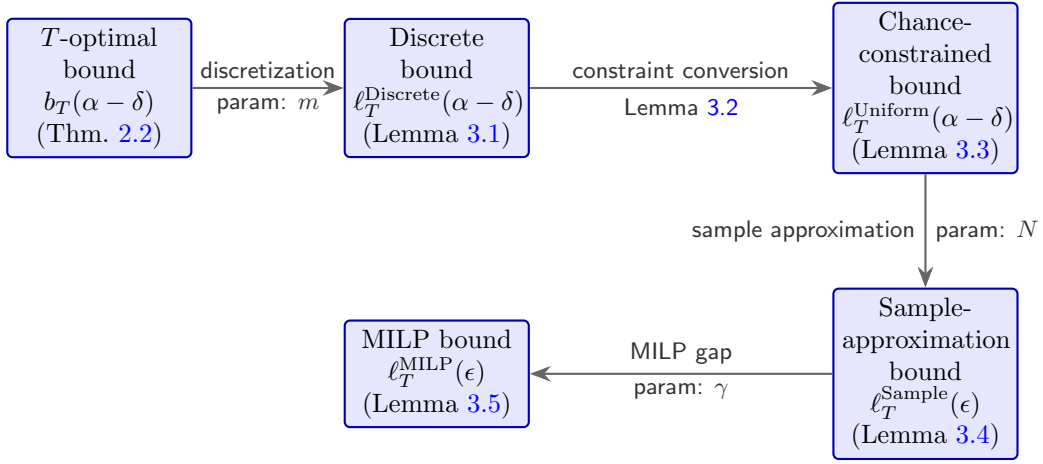

\subsection{Discretization}
\label{sec:discrete_bound}

In this section we assume that the support of distributions in $\F$ is a subset of $ [0,1]$. The bound $b_T(\x)$ in Definition~\ref{def:optimal_bound} requires optimizing over all distributions $F$. To make the problem tractable, we will convert it to optimizing over all distributions on $m$ points  where $m \in \mathbb{N}$ is a parameter of our choice. For any value of $m$, we derive a  lower bound and upper bound of $b_T(\x)$. As $m$ increases, the lower bound and the upper bound converges to $b_T(\x)$, but their computations are more expensive. Therefore $m$ should be chosen to balance the trade-off between computation and performance.

 First we introduce some notations. For a vector $\x \in [0,1]^{m}$, let $x_{(1)}\leq x_{(2)}\leq ...\leq x_{(m)}$ denote the ordered vector. Given a vector $\x \in [0,1]^{m}$ and $0 \le f_{(1)}\leq f_{(2)}\leq ...\leq f_{(m)}=1$, define the following CDF that has CDF $f_{(i)}$ at $x_{(i)}$: 
\begin{align}
    H_{\x, \f}(y) := 
    \begin{cases}
    0 &\text{ if } y \le x_{(1)}\\
    f_{(i)} &\text{ if } x_{(i)} \le y < x_{(i+1)} , 1 \le i \le m-1\\
    1 &\text{ if } y \ge x_{(m)}.
    \end{cases}
\end{align}

We pick a set $\v =(\frac{1}{m}, \cdots, \frac{m-1}{m})$ as the support. Let  $\v \oplus 0 := ( 0, v_1, \cdots, v_{m-1})$ and $\v \oplus 1 := ( v_1, \cdots, v_{m-1}, 1)$. Note that the CDF $H_{\v \oplus 1, F(\v)}$ is always below $F$ and $H_{\v \oplus 0, F(\v)}$ is always above $F$. The discrete distribution $F$ on $m$ points $\v \oplus 1$ is parameterized by $\f \in [0,1]^{m-1}$. We are going to search for $\f$\footnote{Alternatively, we could search for $\v$ instead of searching for $\f$. If we search for $\v$ we do not have to assume that the support is in $[0,1]$ but the MILP will be slightly more complicated (discussed in Section~\ref{sec:MILP}). We leave this to future works.}.

\begin{lemma}[Discretized Bound]
    \label{lem:discrete_bound}
    Let $\F$ be the set of all distributions on support $ [0,1]$. Let $\theta$ be a parameter of distribution.  Suppose that $\theta(F)$ exists for all $F \in \F$. Suppose that $\theta$ satisfies: for any $2$ CDFs $F$ and $G$, if $~\forall x \in \R, F(x) \le G(x)$ then $\theta(F) \ge \theta(G)$. Let $T: \R^n \rightarrow \R$ satisfy the following condition: if CDF $G$ first-order stochastically dominates CDF $F$, then $\P_F(T(\X) \ge T(\x)) > \alpha$ implies $\P_{G}(T(\X) \ge T(\x)) > \alpha$. Let $m \in \mathbb{N}, m \ge 1$. 
    
 Let $C(F, \alpha)$ denote the constraint: $\P_{F} ( T(\X) \ge T(\x)) > \alpha$. 
 
Let $\Delta_m:=\sup_{\f \in [0,1]^m, f_m = 1} (\theta(H_{\v \oplus 1, \f})  - \theta(H_{\v \oplus 0, \f})).$ 

For any $\x \in \R^n$,  the upper bound and lower bound are defined as: 
\begin{align}
   & u^{\text{Discrete}}_T( \x, \alpha, m) \\
   &~~~~~~~~:= \begin{cases} \inf_{\f \in [0,1]^m, f_m = 1} \theta(H_{\v \oplus 1 , \f}) \text{ subject to } C(H_{\v \oplus 1 , \f} , \alpha) &\text{ if the problem is feasible} \\
    \infty &\text{ if the problem is infeasible} \end{cases}\\
    &\ell^{\text{Discrete}}_T(\x, \alpha, m) := u^{\text{Discrete}}_T( \x, \alpha, m) - \Delta_m. 
\end{align}

Let $b^{\F}_T(\x)$ the $T$-optimal bound of $\theta$. Then
\begin{align}
&\ell^{\text{Discrete}}_T(\x, m) \le b^{\F}_T(\x) \le u^{\text{Discrete}}_T(\x, m )\\
&u^{\text{Discrete}}_T(\x, m ) - b^{\F}_T(\x) \le \Delta_m \\
&b^{\F}_T(\x) - \ell^{\text{Discrete}}_T(\x, m) \le \Delta_m
\end{align}
and $\Delta_m = \frac{1}{m}$ if $\theta$ is the mean of the distribution. 
\end{lemma}
We will omit some of the inputs and subscript of $ u^{\text{Discrete}}_T( \x, \alpha, m)$ and $ \ell^{\text{Discrete}}_T( \x, \alpha, m)$ when the context is clear.

If $T$ is monotonic then $T$ satisfies the condition in Lemma~\ref{lem:discrete_bound}. If $T$ does not satisfy the condition in Lemma~\ref{lem:discrete_bound}, as $m$ increases the discrete CDF on $\v$ might converge to the true CDF $F$ and therefore the discrete bound $\ell$ might still approximate $b$ well. We leave this extension to future works. 

\subsection{Conversion to Chance-constrained Programs}
\label{sec:chance_constraint}
Consider the problem of computing $\ell^{\text{Discrete}}_T(\x, m )$ in Lemma~\ref{lem:discrete_bound}. Let $H:= H_{\v\oplus 1, \f}$, which is a discrete distribution on support $\v\oplus 1$. Then 
\begin{align}
P_H(T(\X) \ge T(\x)) = \sum_{\y \in (\v\oplus 1)^n, T(\y) \ge T(\x)} \Pi_{i=1}^n \P_H(X_i = y_i).
\end{align}

The sum could potentially has $\Omega(m^n)$ elements. If we assume that the support $\v$ is known and want to find $f_k = \P(X_i = v_k)$, then the term on the right-hand side (RHS) has degree $n$ in term of $\f$. If we assume that $\f$ is known and want to find the quantiles $\v$, then the term on the RHS contains integer constraints. In both cases the RHS is a term with exponential size and is non-linear or has integer constraints, which makes the optimization problem difficult to solve. 

To resolve this problem, we will  convert the constraint $\P(T(\X) \ge T(\x))$ to chance constraints \citep{Kim2015, nemirovski2006convex, Pagnoncelli2009Sample, luedtke2007sample} so that we can  obtain a MILP with size linear in $m$ and $n$.

\begin{lemma}
\label{lem:constraint_equal}
 Let $T: \R^n \rightarrow \R$ be a statistic of a sample of size $n$. Let $F$ satisfy $\P_F(X \le 1) = 1$. Let $\U$ be $n$ i.i.d samples from the half-open interval $(0,1]$. Then for any $\x \in \R^n$:
\begin{align}
\P_{F}(T(\X) \ge T(\x)) \le \P_{\U} (\exists \y \in (-\infty, 1]^n: T(\y) \ge T(\x),  \U \leq F^{\geq}(\y)) .
\end{align}

If $T$ is monotonic, then for any $\x \in \R^n$:
\begin{align}
\P_{F}(T(\X) \ge T(\x)) =\P_{\U} (\exists \y \in (-\infty, 1]^n: T(\y) \ge T(\x),  \U \leq F^{\geq}(\y)) .
\end{align}
\end{lemma}
We can sample $\U$ to approximate the RHS, where the unknown parameter $F$ appears in linear constraints $\U \leq F^{\geq}(\y)$.

From Lemma~\ref{lem:constraint_equal}, the lower bound in Lemma~\ref{lem:discrete_bound} is equivalent to the following bound:
\begin{lemma}[Chance-constrained Program]
\label{lem:chance_constrained_bound}
  Let $\F$ be the set of all distributions on support $ [0,1]$. Let $\theta$ be a parameter of distribution.  Suppose that $\theta(F)$ exists for all $F \in \F$. Let $T: \R^n \rightarrow \R$ be a statistic.  Let $m \in \mathbb{N}, m \ge 1$.

Let $\P_{\U(0,1]}$ denote the probability where $\U$ is sampled from the half-open interval $(0,1]$. Let $C^{\text{Uniform}}(F, \alpha)$ denote the constraint: 
\begin{align}
\P_{\U(0,1]} (\exists \y \in (-\infty, 1]^n: T(\y) \ge T(\x),  \U \leq F^{\geq}(\y))  > \alpha. 
\end{align}

For any $\x \in \R^n$, define:  
\begin{align}
&\ell^{\text{Uniform}}_T( \x,  \alpha, m) \\
&:= \begin{cases}\inf_{\f \in [0,1]^m, f_m = 1} \theta(H_{\v \oplus 1 , \f}) - \Delta_m \text{ subject to } C^{\text{Uniform}}(H_{\v \oplus 1 , \f} , \alpha)  &\text{ if the problem is feasible} \\
\infty &\text{ if the problem is infeasible.} \end{cases}
\end{align}

Let $\ell^{\text{Discrete}}_T(\x, \alpha, m)$ be the discretized bound defined in Lemma~\ref{lem:discrete_bound}.  Then
\begin{align}
 \ell^{\text{Discrete}}_T( \x,  \alpha, m) \le \ell^{\text{Uniform}}_T( \x,  \alpha, m).
\end{align}

If $T$ is monotonic, then
\begin{align}
    \ell^{\text{Discrete}}_T( \x, \alpha, m) = \ell^{\text{Uniform}}_T( \x,  \alpha, m).
    \end{align}
    \end{lemma}
We will omit some of the inputs and subscript of  $ \ell^{\text{Uniform}}_T( \x, \alpha, m)$ when the context is clear.

\subsection{Sample Approximation}
\label{sec:sample_approximation}
We use the discretization from Section~\ref{sec:discrete_bound} and the chance constraints from Section~\ref{sec:chance_constraint} to obtain a MILP with size linear in $m$ and $n$. We compute a lower bound $\ell^{\text{Sample}}_T(\x, \epsilon, m,  N, \U)$ of the the chance-constrained program $\ell^{\text{Uniform}}_T$ using sample approximation \citep{luedtke2007sample}.

We first sample $N$ samples of $\U^{i} \in (0,1]^n$ so that the constraint $C^{\text{Uniform}}(H_{\v \oplus 1 , \f}, \alpha )$  is true becomes $C^{\text{Sample}}(H_{\v \oplus 1 , \f}, \epsilon, N)$. 
\begin{lemma}[informal]
\label{lem:sample_approximation_bound}
    Let $\U^{i} \in (0,1]^n, ~1 \le i \le N$ be $N$ i.i.d. samples of vectors of $n$ i.i.d. samples of uniform random variables. Let $C^{\text{Sample}}(F, \epsilon, N, \U )$ define the constraint:  
    \begin{align}
  \frac{1}{N}\sum_{i=1}^N \I\left[  \exists \y^i \in  (-\infty, 1]: T(\y^i) \ge T(\x),      \U^i \leq F^{\geq}(\y^i)  \right] \ge  \epsilon . 
\end{align}
For any $\x \in \R^n$, define
\begin{align}
 &\ell^{\text{Sample}}_T(\x, \epsilon, m, N, \U) \\ &:= \begin{cases}\inf_{\f \in [0,1]^m, f_m = 1} \theta(H_{\v \oplus 1 , \f}) \text{ subject to } C^{\text{Sample}}(H_{\v \oplus 1 , \f} , \epsilon, N, \U) &\text{ if the problem is feasible}\\
 \infty &\text{ if the problem is infeasible.}\end{cases} 
 \end{align}
 Let $\ell^{\text{Uniform}}_T(\x, \alpha, m)$ be the chance-constrained bound defined in Lemma~\ref{lem:chance_constrained_bound}. Then $\ell^{\text{Sample}}_T(\x, \epsilon, m , N, \U)$ with $\epsilon \le \alpha$ is a lower bound of ${\ell}^{\text{Uniform}}_T(\x, \alpha,m)$ with high probability. 
\end{lemma}
 We omit some of the inputs and subscripts of $\ell^{\text{Sample}}_T(\x, \epsilon, m,  N, \U)$ when the context is clear.

Using the result in \citep{{luedtke2007sample}} we discuss how to choose the value of $N$ and $\epsilon$  so that the sample-approximation result $\ell^{\text{Sample}}_T(\x, \epsilon, m,  N)$ is a lower bound of the chance-constrained program ${\ell}^{\text{Uniform}}_T(\x, \alpha - \delta, m)$ with high probability $1-\delta$ where $\alpha > \delta > 0$. We compute the chance-constrained program ${\ell}^{\text{Uniform}}_T(\x, \alpha - \delta, m)$ with significance level $\alpha - \delta$ instead of $\alpha$ to compensate for the $\delta$ probability that $\ell^{\text{Sample}}$ is not a lower bound of $\ell$. There are $2$ options with the proofs in Appendix~\ref{app:sample_approximation}:
\begin{itemize}
   \item Solve the problem in Lemma~\ref{lem:sample_approximation_bound} only $M=1$ time with $0 < \epsilon < \alpha - \delta$ and $N$ samples. If
\begin{align}
N \ge \frac{1}{2(\alpha- \delta -\epsilon)^2}\ln\frac{1}{\delta}, 
\end{align}
then \begin{align}
\P_{\U}(\ell^{\text{Sample}}_T(\x, \epsilon, m, N, \U) \le {\ell}^{\text{Uniform}}_T(\x, \alpha- \delta, m)) \ge 1- \delta. 
\end{align}
For example, if $\alpha = 0.05, \delta = 0.001, \epsilon = 0.03$ then $N = 9568$. If $\alpha = 0.05, \delta = 0.001, \epsilon = 0.048$ then $N = 3.5$ millions.  Note that $N$ is derived from Hoeffding's inequality, so it could be overly conservative. Using a tighter inequality can make it smaller. 
\item Choose $\epsilon$ such that $0 < \epsilon \le \alpha - \delta$ and $\epsilon N$ is an integer.  Generate and solve $M$ the problem in Lemma~\ref{lem:sample_approximation_bound} independently, each with with $N$ samples and $\epsilon$, and then take the smallest result among $M$ results, denoted $\underline\ell^{\text{Sample}}_T(\x, \epsilon,m, N, \U, M)$\footnote{Gemini contributed to the derivation of this formulation.}. Note that for each of the $M$ problems, we need to resample $N$ samples of $n$ i.i.d. uniform random variables separately (see Appendix~\ref{app:sample_approximation} for the details). Then if  $M \ge \log_2\left( \frac{1}{\delta}\right)$  and $\epsilon \le \alpha - \delta$ then
\begin{align}
\P_{\U}(\underline\ell^{\text{Sample}}_T(\x, \epsilon,m, N, \U, M) \le {\ell}^{\text{Uniform}}_T(\x, \alpha- \delta, m)) \ge 1 - \delta. 
\end{align}
In this option $N$ can have any value, but in practice we want $N$ to be large so that the bound is tighter. 
If $\delta = 0.001$ then $M =10$. 
\end{itemize}
To simplify notations, we use $\underline\ell^{\text{Sample}}_T(\x, \epsilon,m, N, \U, M)$ to denote the final result in either option, and $\ell^{\text{Sample}}_T(\x, \epsilon,m, N, \U(s))$ to refer to the result of solving a single problem $s$ among $M$ problems. 

 \subsection{Mixed-Integer Linear Program}
\label{sec:MILP_bound}
The following transformations turn the sample-approximation problem $\ell^{\text{Sample}}$ with the constraint $C^{\text{Sample}}(H_{\v \oplus 1 , \f} , \epsilon, N) $ into a MILP. The full algorithm is in Appendix~\ref{app:complete_algo}. To simplify notations, in this section we use $\ell^{\text{Sample}}_T(\x, \epsilon,m, N, \U)$ to refer to the result of solving a single problem in Lemma~\ref{lem:sample_approximation_bound} and $\U\in (0,1]^{N \times n}$ to be $N$ samples $\U^i \in (0,1]^n$ of $n$ i.i.d. uniform random variables.   The output of the MILP solver is denoted $\ell^{\text{MILP}}(\x, \epsilon, m,  N, \U, \gamma)$, which is an approximation of $\ell^{\text{Sample}}_T$ due to the MILP gap $\gamma$. 
\begin{itemize}
\item \textbf{Eliminating $\I$: }
Using big M method for counting, $\sum_{i=1}^N \I[x_i \ge y_i] \ge k$ is equivalent to 
\begin{align}
&\forall i, 1 \le i \le N,  z_i \in \{0,1\}; ~\sum_{i=1}^N z_i \ge k; ~~\forall i, 1 \le i \le N: x_i \ge z_i y_i. 
\end{align}
If $z_i = 1$, then $x_i \ge y_i$. Therefore if $\sum_{i=1}^N z_i \ge k$ then there are at least $k$ value of $x_i \ge y_i$. The constraint $C^{\text{Sample}}(H_{\v \oplus 1 , \f} , \epsilon, N) $ is equivalent to
\begin{align}
&\forall i, 1 \le i \le N,  z_i \in \{0,1\}\\
&\sum_{i=1}^N z_i \ge  \lceil N \epsilon \rceil\\
&\forall i, 1 \le i \le N, \exists \y^i \in (-\infty, 1]^n: T(\y^i) \ge z_i T(\x) ,  ~H^{\geq}_{\v \oplus 1, \f}(\y^i) \ge z_i \U^i . 
\end{align}
\item \textbf{Eliminating $\exists$: }
The constraint $\exists \y^i \in (-\infty, 1]^n: T(\y^i) \ge z_i T(\x) , ~ H^{\geq}_{\v \oplus 1, \f}(\y^i) \ge z_i \U^i $ is equivalent to the $3$ constraints: 
\begin{align}
\y^i \in (-\infty, 1]^n; ~~T(\y^i) \ge z_i T(\x); ~~H^{\geq}_{\v \oplus 1, \f}(\y^i)  \ge z_i \U^i .
\end{align}

\item \textbf{Reducing the scope of $\y$:} 
Let $v_m:=1$. By definition, 
\begin{align}
H^{\geq}_{\v \oplus 1, \f}(y) = \begin{cases}
1 &\text{ if } y \le v_1 \\
1- f_{k} &\text{ if } v_{k} < y \le v_{k+1}, ~\forall k, 1 \le k \le m-1. 
\end{cases}
\end{align}
From the property that if $\x \preceq \y$ then $T(\x) \le T(\y)$, the objective is optimal if $y^i_j \in \{v_1, \cdots, v_m\}$ because the solver will push $y^i_j$ in the interval $(v_k, v_{k+1}]$ with the same survival function to be as large as possible to satisfy $ T(\y^i) \ge T(\x) $. Therefore $\y^i \in (-\infty, 1]^n$ is equivalent to $\y^i \in \{v_1, \cdots, v_m\}^n$\footnote{Gemini contributed to the derivation of this formulation.}.

\item \textbf{Eliminating the survival function $H^{\geq}_{\v \oplus 1, \f}(\y^i)$: }
We now show how to express $H^{\geq}_{\v \oplus 1, \f}(\y^i) $ as a set of linear constraints, therefore the problem becomes a Mixed-Integer Linear Programming except for the constraint $T(\y^i) \ge T(\x)$.  Let $w_{i,j,k}$ be defined as: 
\begin{align}
w_{i,j,k} := \begin{cases}
0 &\text{ if } y^i_j > v_{k-1} \\
1 &\text{ if } y^i_j \le v_{k-1}. 
\end{cases}
\end{align}
Then if $y^i_j = v_k$, then $w_{i,j,1} = \cdots = w_{i,j, k} = 0$ and $w_{i,j,k+1} = \cdots = w_{i,j, m} = 1$. Therefore\footnote{The formula to compute $y^i_j$ includes the product $w_{i,j,k} v_k$. If we search for $\v$ instead of search for $\f$, we will need to use the big $M$ method to convert the product to linear functions.}
\begin{align}
 y^i_j = 1 - \frac{1}{m}\sum_{k=1}^{m} w_{i,j,k}. 
\end{align}

The constraint $ H^{\geq}_{\v \oplus 1, \f}(\y^i)  \ge z_i \U^i$ is equivalent to\footnote{Gemini contributed to the derivation of this formulation.}
\begin{align}
1 - f_{k-1} \ge  U^i_j(z_i - w_{i,j,k} ). 
\end{align}

\end{itemize}

In the analysis below we assume the case  where we use the MILP to maximize $1-\theta$ instead of minimizing $\theta$. Let $r$ be the upper bound of the objective $1-\theta$ that the MILP solver outputs (Appendix~\ref{app:MILP}). Then we calculate $u = 1- r$, and $ \ell^{\text{MILP}} = u - \Delta_m$. 
 
\begin{lemma}[MILP Bound]
\label{lem:milp_bound}

Let $\F$ be the set of all distributions on support $ [0,1]$. Let $\theta$ be a parameter of distribution. Suppose that $\theta(F)$ exists for all $F \in \F$. Let $T: \R^n \rightarrow \R$ be a monotonic statistic. Let $m \in \mathbb{N}, m \ge 1$. 

    Let $\ell^{\text{MILP}}_T(\x,  \epsilon, m, N,\U, \gamma)$ denote the bound computed by the MILP in Section~\ref{app:complete_algo} that maximizes $1-\theta$ with mixed-integer programming (MIP) gap $\gamma \ge 0$ if the MILP is feasible and $\infty$ if the MILP is infeasible.  Then
    \begin{align}
    \ell^{\text{MILP}}_T(\x, \epsilon,m, N, \U, \gamma) \le \ell^{\text{Sample}}_T(\x, \epsilon, m , N, \U) \le 1- \frac{1-\ell^{\text{MILP}}_T(\x, \epsilon,m, N,\U, \gamma)}{1+ \gamma} . 
    \end{align}
\end{lemma}
We will omit some of the inputs and subscript of  $ \ell^{\text{MILP}}_T(\x, \epsilon,m, N, \U, \gamma)$ when the context is clear. If we know a tighter set of support (for example $\cX = [0,1/2] \cup [3/4, 1]$ instead of $[0,1]$), we could add more constraints on $\f$ so that the CDF of $H_{\v \oplus 1, \f}(y^i_j)$ is flat in the region not in $\cX$ and therefore making the optimization tighter.

\subsection{Analysis}
\label{sec:analysis}

\begin{theorem}
\label{thm:analysis_extend}
Let $\F$ be the set of all distributions on support $ [0,1]$. Let $\theta$ be a parameter of distribution.  Suppose that $\theta(F)$ exists for all $F \in \F$. Suppose that $\theta$ satisfies: for any $2$ CDFs $F$ and $G$, if $~\forall x \in \R, F(x) \le G(x)$ then $\theta(F) \ge \theta(G)$. Let $T: \R^n \rightarrow \R$ be a monotonic statistic. Let $m \in \mathbb{N}, m \ge 1$. Let $\delta \in (0,\alpha)$. Let  $\epsilon \le \alpha-\delta$ be the significance level of the MILP and  $\gamma \ge 0$ be the MIP gap. Let $\Delta_m:=\sup_{\f \in [0,1]^m, f_m = 1} (\theta(H_{\v \oplus 1, \f})  - \theta(H_{\v \oplus 0, \f})).$ 

Let $\ell^{\text{MILP}}_T(\x, \epsilon,m, N,\U(s), \gamma) , 1 \le s \le M$ be the output of generating and solving MILP $s, ~1\le s \le M$ in Appendix~\ref{app:complete_algo} where the MILP maximizes $1-\theta$ and $\U(s) \in (0,1]^{N \times n}$. Let $\ell^{\text{Sample}}_T(\x, \epsilon,m, N,\U(s), \gamma) , 1 \le s \le M$ be the results of the corresponding problem in Lemma~\ref{lem:sample_approximation_bound} that MILP $s, ~1\le s\le M$ approximates. Let 
\begin{align}
&\underline\ell^{\text{MILP}}_T(\x, \epsilon,m, N,\U, M, \gamma) = \min_{s, 1 \le s \le M} \ell^{\text{MILP}}_T(\x, \epsilon,m, N,\U(s), \gamma) \\
&\underline\ell^{\text{Sample}}_T(\x, \epsilon,m, N,\U, M) = \min_{s, 1 \le s \le M} \ell^{\text{Sample}}_T(\x, \epsilon,m, N,\U(s)). 
\end{align}

Then
\begin{align}
& ~\forall s, 1 \le s \le M: \\&~~~~~~~~~~\ell^{\text{MILP}}_T(\x, \epsilon,m, N,\U(s), \gamma) \le \ell^{\text{Sample}}_T(\x, \epsilon, m , N,\U(s)) \le 1- \frac{1-\ell^{\text{MILP}}_T(\x, \epsilon,m, N,\U(s), \gamma)}{1+ \gamma} \\
&\P_{\U}(\underline\ell^{\text{Sample}}_T(\x, \epsilon, m , N,\U, M) \le {\ell}^{\text{Uniform}}_T(\x, \alpha-\delta,m)) \ge  1 - \delta \\
&  {\ell}^{\text{Uniform}}_T(\x, \alpha-\delta,m) =  \ell^{\text{Discrete}}_T(\x, \alpha-\delta,m)\\
& \ell^{\text{Discrete}}_T(\x, \alpha - \delta,m ) \le b^{\cF}_T(\x, \alpha - \delta) \le \ell^{\text{Discrete}}_T(\x, \alpha - \delta) + \Delta_m  \\
&\P_{\X}({b}^{\cF}_T(\X, \alpha - \delta) \le \theta) \ge  1 - (\alpha - \delta). 
\end{align}
From which we can conclude that
\begin{align}
    \P_{\U, \X} (\underline\ell^{\text{MILP}}_T(\X, \epsilon, m,  N, \U, M, \gamma) \le \theta) \ge 1 - \alpha. 
\end{align}
\end{theorem}

Note that $\ell^{\text{MILP}}$ and $\ell^{\text{Discrete}}_T$ can be made arbitrarily close to the targets they approximate by decreasing $\gamma$ and increasing $m$. If we can make the sample-approximation result $\ell^{\text{Sample}}_T$ arbitrarily close to the chance-constrained result $\ell^{\text{Uniform}}_T$ (which we leave to future works) then it might be possible that $\ell^{\text{MILP}}$ will be arbitrarily close to ${b}^{\cF}_T(\x, \alpha - \delta)$ (Figure~\ref{fig:process}). We leave this extension to future works.

\section{Experiments}
\label{sec:exp}
\begin{figure}[htbp]
  \centering
  \begin{subfigure}[t]{0.45\textwidth}
    \includegraphics[width=\textwidth]{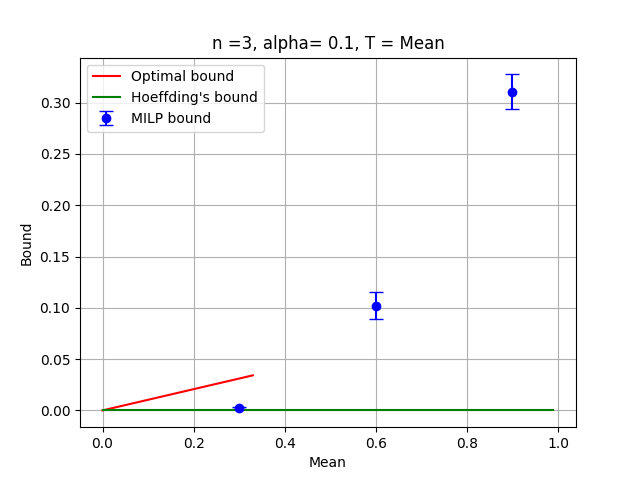}
    \caption{\label{fig:Tmean} $n=3$. The sample mean is the $x$-axis and the order function. The MILP bound is better than Hoeffding's bound. The MILP bound is a lower bound of the $T$-optimal bound when the sample mean is $0.3$. As discussed in Section~\ref{sec:analysis}, if the sample-approximation result converges to the chance-constrained result as $N$ increases (which we leave to future works), then increasing the parameters $N, m$ and decreasing $\gamma$ should make the MILP bound converge to the $T$-optimal bound (red line).  }
  \end{subfigure}
  \hfill
  \begin{subfigure}[t]{0.45\textwidth}
    \includegraphics[width=\textwidth]{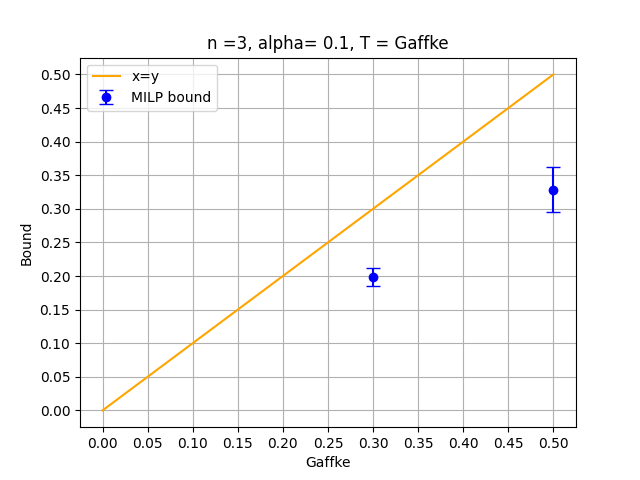}
    \caption{\label{fig:Tgaffke} $n=3$. The approximate Gaffke's bound with $N' = 100$ is the $x$-axis and the order function. Since the approximate Gaffke's bound is not yet proven to be valid, the optimal bound may or may not lie above the diagonal. }
  \end{subfigure}
    \hfill
  \begin{subfigure}[t]{0.45\textwidth}
    \includegraphics[width=\textwidth]{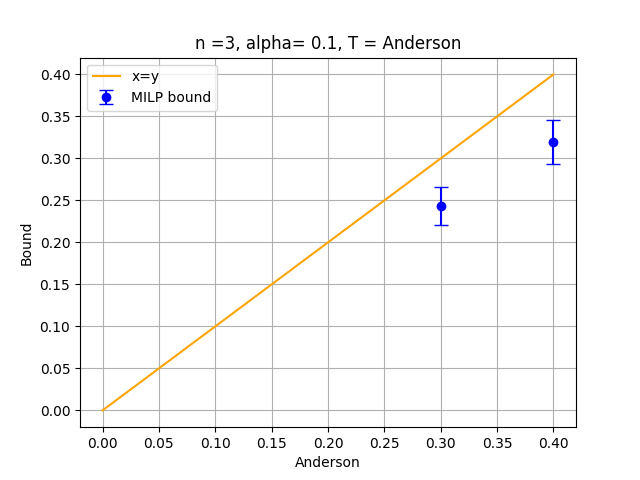}
    \caption{\label{fig:Tanderson} $n=3$. Anderson's bound is the $x$-axis and the order function.  Since Anderson's bound is valid, the $T$-optimal bound should be equal to or larger than Anderson's bound (at or above the diagonal $x=y$ line). As discussed in Section~\ref{sec:analysis}, if the sample-approximation result converges to the chance-constrained result as $N$ increases (which we leave to future works), then increasing the parameters $N, m$ and decreasing $\gamma$ should make the MILP bound converge to the $T$-optimal bound (at or above the diagonal line).  }
  \end{subfigure}
  \hfill  
  \begin{subfigure}[t]{0.45\textwidth}
    \includegraphics[width=\textwidth]{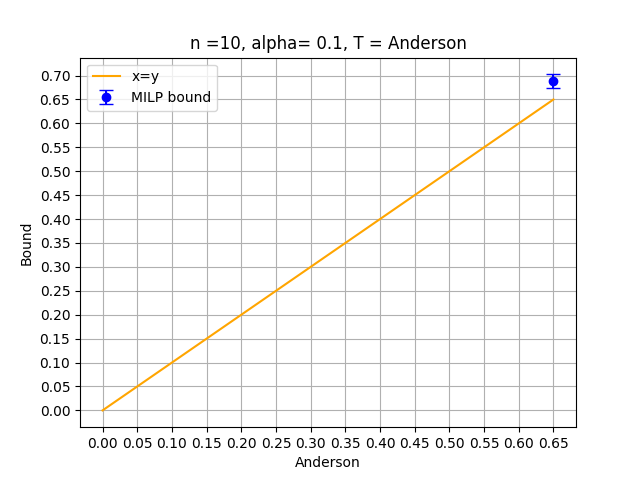}
    \caption{\label{fig:tune_anderson} Tuning Anderson's bound when $n=10$. Anderson's bound is the $x$-axis and the order function.  Our approximate MILP bound is better than Anderson's bound (above the diagonal $x=y$ line). Therefore it can be used to tune Anderson's bound.}
  \end{subfigure}
  \caption{Experiments with small parameters when the order function is the sample mean, Anderson's bound and Gaffke's bound.  We compute the bounds (denoted "MILP bound") for $n_{\text{exps}}=10$ times and and plot the mean and standard deviation. The standard deviations are significant, suggesting that increasing the size of $N$ will make the results over $10$ times more concentrated and the MILP bound increase.}
\end{figure}
In this section we compute lower confidence bounds of the distribution's mean using the algorithm in Appendix~\ref{app:complete_algo}.\footnote{Gemini generated and optimized parts of the code.} We perform small experiments with $n = 3, \alpha = 0.1,  \delta = 0.001, N=100, M=10, m = 50$ and the MIP gap  $\gamma = 0.02$. The details of the implementation is in the Appendix. With these parameters, we solve the MILP in CVXPY with the HIGHS solver, which takes minutes in a personal laptop. We compute the bound for $n_{\text{exps}} = 10$ times for each $T(\x)$, and plot the mean and standard deviations of $n_{\text{exps}}$ values. To compute the bound each time, we solve $M = 10$ MILPs with $N = 100$, $\epsilon = \frac{\lfloor N (\alpha - \delta) \rfloor}{N} = 0.09, m=100$ and the MIP gap $\gamma = 0.01$ to maximize $\frac{1}{m} \sum_{i=1}^{m-1}f_i$. Let $r_1, \cdots, r_M$ be the output of the $M =10$ MILPs.\footnote{This value of $M$ is conservative. It is enough that $M = \lceil \log_2\frac{1}{\alpha - \epsilon}\rceil = 7$.} maximizing $\frac{1}{m} \sum_{i=1}^{m-1}f_i$. Then the guaranteed lower bound is the minimum over $1 - r_1 - 1/m, \cdots, 1- r_M - 1/m $ of the $M$ MILP results. 

To speed up the solver, we add some additional constraints that do not change the problem including the constraint $\forall i, j, k: w_{i,j,k} \le z_i$. We also set $z_i$ to be continuous to reduce the number of discrete variables. Since we have the constraint $\forall i, j, k: w_{i,j,k} \le z_i$, if $w_{i,j,k} =1$ then $z_i = 1$ so this is not a significant relaxation. The details are in Appendix~\ref{app:exp}.

Note that when we increases $N$ and $m$ and decrease $\gamma$, the bound will increase but the computation time will be longer. Using a commercial solver such as Gurobi will also increase the speed. For each value of $T(\x)$, we only need to compute the bound once and save the result in a table. To use the bound with a new sample $\x$, we can compute $T(\x)$ and refer to the table to find the value of the bound. 

Solving a single MILP in all of the experiments in this section usually takes minutes on a personal laptop. When $T(\x)$ is large, the solver is much faster than when $T(\x)$ is small. 

The order functions in this section (the sample mean, Anderson's bound and Gaffke's bound) are monotonic.

\subsection{The Sample Mean as the Order Function}
We compute the bound when $T(\x)$ is the sample mean of $\x$. The constraint $T(\y^i) \ge T(\x)$ is a linear constraint. In Figure~\ref{fig:Tmean} we calculate the bound when $T(\x) = 0.3, 0.6$ and $0.9$. From Theorem \ref{thm:luczak_positive}, if $\frac{1}{n}\sum_i x_i \le \frac{1}{n}$ , $n=3$ and $ \alpha < 0.63$ then the $T$-optimal bound $b^{[0, 1]}(\x) = (\sum_{i=1}^n x_i) (1 - (1-\alpha)^{1/n}).$ We plot the $T$-optimal bound in red when $T(\x) \le 1/n$ and Hoeffding's for all $T(\x) \in [0,1]$. When $T(\x) =0.3$, our bound is a lower bound of the $T$-optimal bound. Our bound is better than Hoeffing's. 

Note that there are other tail bounds such as \citep{bentkus} which could be used to derive sample-mean-ordered bounds that could be tighter than Hoeffding's. We do not provide comparisons with these bounds here. 
\subsection{Gaffke's Bound as the Order Function}

 Gaffke's bound \citep{Gaffke2005} is an upper confidence bound of the mean of distributions on $[0,1]$ with good performance in experiments \citep{Learned-MillerThomas2019} but its validity has not been proven. The experiments on Gaffke's bound were performed using an approximation of the bound , which we call the approximate Gaffke's bound. We find the lower bound of the mean of $X \in [0,1]$ by finding the upper bound of the mean of $1-X \in [0,1]$. We sample $N'$ samples of $n$ i.i.d uniform random variables $\a^{i'} \in [0,1)^n$. The approximate Gaffke's bound $T_{\a}(\x)$ is computed through the procedure: 

\begin{enumerate}
    \item Compute $m_{i'}:= \mu(H_{(1-\x) \oplus 1,\a^{i'}})$ for every $i'$. 
    \item Sort the $N'$ value computed in the previous step so that $m_1 \le m_2 \le \cdots m_{N'}$. 
    \item Let $q:= \lceil N' (1-\alpha) \rceil$. Let $\hat{Q}_{\a}(1-\alpha,1- \x):= m_q$. 
\item Output $T_{\a}( \x) = 1- \hat{Q}_{\a}(1-\alpha, 1-\x)$.
\end{enumerate}

 We treat $\A = \a$ as a constant vector. At the beginning, we sample $N'$ samples of $n$ i.i.d. uniform random variables $\A^{i'} \in [0,1]^n$ and then use the sampled values $\a$ as constants to compute $T_{\a}(\x)$ for any sample $\x$. We use $N' = 100$ to sample $\a$ for the approximate Gaffke's bound. Note that the MILP bound is valid for any $N'$, but larger $N'$ will make the approximate Gaffke's bound closer to the Gaffke's bound. For the value of $\a$ we sampled, the largest $T(\x)$ when $\x =1$ is larger than $0.52$. Therefore we calculate the MILP bound when $T(\x) = 0.3$ and when $T(\x) = 0.5$. 
We plot the result in Fig.~\ref{fig:Tgaffke}. 
\subsection{Anderson's Bound as the Order Function} 
Anderson's bound \citep{Anderson1969} is a linear function of the sorted sample and is a valid bound proven to be tighter than Hoeffding's for any sample. Therefore the constraint $T(\y^i) \ge T(\x)$ is a linear constraint.  

Anderson's upper bound of the mean of distribution on $(-\infty, 1]$ is derived by derived a lower bound of the CDF $F$ using the property $\P(F(X) \ge y)  = 1 - \P(F(X) < y) \ge 1 - \P(F(X)\le y) \ge 1 - y =  \P(U \ge y).$ Since it's not always true that  $\P(F(X) \ge y)\le  \P(U \ge y)$, we cannot derive a lower bound of the mean by finding an upper bound of the CDF. Therefore, we derive a lower bound of the mean of $X \ge 0$ by finding an Anderson's upper bound of the mean of $1-X \le 1$. 

The valid bound in \citep{phan2021practical} is strictly tighter than Anderson's bound when applied to applied to the set of all distributions with support on $[0,1]$. Therefore the $T$-optimal bound with Anderson's as the order function will be strictly tighter than Anderson's bound when applied to applied to the set of all distributions with support on $[0,1]$. From Section~\ref{sec:analysis}, the chance-constrained bound $\ell^{\text{Uniform}}_T$ should be within $1/m$ of the $T$-optimal bound, and therefore should be strictly less than $1/m$ worse than Anderson's bound. However we are approximating the sample-approximation bound $\ell^{\text{Sample}}$ by $\ell^{\text{MILP}}$ and we do not know how close $\ell^{\text{Sample}}$ is to the chance-constrained bound $\ell^{\text{Uniform}}$ (as discussed in Section~\ref{sec:analysis} it is left to future works). 

We use the calculation of Anderson's bound through Monte Carlo simulations from \citep{phan2021practical}.  We compute the coefficients $\c$ only once at the beginning, and then use it as constant coefficients to compute $T(\x)$ for any sample $\x$. Let Anderson's bound be $T(\x) = \sum_{i=1}^n c_i x_i$ where $x_1 \le x_2 \le \cdots \le x_n$ is sorted. With $n= 3$ and $\alpha = 0.1$ we calculate $\c =[0.33333333, 0.10205606, 0.        ]$. The largest possible value of $T(\x)$ when $\x = 1$ is larger than $0.43$. Therefore we calculate the bound for when $T(\x) = 0.2$ and $T(\x) = 0.4$.  We plot the result in Figure~\ref{fig:Tanderson}. 

Note that our bound is valid for any linear coefficient $c_i$ of $T(\x)$. However if the Monte Carlo coefficient converges to Anderson's bound, then then the Monte Carlo approximation of Anderson's bound will be approximately valid, and the result in Theorem~\ref{thm:optimal} will be approximately applicable. 

  \textbf{Tuning Anderson's bound.}  
    We also perform another experiment in Figure~\ref{fig:tune_anderson} to see if our bound can be better than Anderson's bound. We let $n=10$ and other parameters stay the same. With $n = 10, \alpha = 0.1$ we calculate $\c = [0.1       , 0.1       , 0.1       , 0.1       , 0.1       ,
       0.1       , 0.07735639, 0.        , 0.        , 0.        ]$.   The maximum value of $T(\x) = \c^T \x$ when $\x = 1$ is larger then $0.67$. Therefore we can let $T(\x) = 0.65$. 
Our MILP bound is larger than $0.65$.

\section*{Acknowledgment}

We would like to thank Philip B. Stark and Jacob Spertus for the discussions. We used large language models to assist in  deriving proofs,  optimizing the formulation of the MILP and  generating and optimizing the codes (some of the specific instances are described in the footnotes). Large language models' suggestions have been double checked. 

\bibliography{mean_interval}
\bibliographystyle{plainnat}  
\clearpage
\appendix

In Section \ref{app:complete_algo} we show the complete algorithm and discuss the option to use inverting test hypotheses methods (e.g. \citep{waudbysmith2021estimating}) as the order function (Section~\ref{app:T_as_inverse}).

In Section~\ref{app:exp} we discuss experiment details. 

In Section~\ref{app:proof} we show the proofs, including the error and sample complexity analysis in Section~\ref{sec:analysis}.

In Section~\ref{app:connection} we show the connections to pivoting the CDF  and inverting hypothesis tests \citep{CaseBerg:01}.

In Section~\ref{app:related_works} we discuss related works.

\section{The Complete Algorithm}
\label{app:complete_algo}
\begin{enumerate}
\item Inputs: Sample size $n$, sample $\x \in \R^n$, monotonic function $T: \R^n \rightarrow \R$,  thresholds $\alpha, \epsilon, \delta$  such that $0 < \alpha < 1, ~0< \delta < \alpha,~ 0 < \epsilon \le \alpha - \delta$, MIP gap $ \gamma \ge 0$, natural positive numbers $m, N  \in \mathbb{N}, m, N \ge 1$. The meaning of the parameters are noted below: 
\begin{itemize}
\item $\delta$ is to account for the probability that the sample approximation solution is not a lower bound of the chance-constrained program. To account for this failure event, we find a $\alpha-\delta$ bound instead of $\alpha$, and use $\epsilon \le \alpha-\delta$ in the MILP. The bound is tighter if $\delta$ is small and $\epsilon$ is large and close to $\alpha-\delta$. $\delta$ should be chosen closer to $0$, such as $0.001$. 

    \item $N$ is the number of samples to approximate the chance-constrained programs. We use $N$ samples of $n$ i.i.d random variables $\U \in (0,1]^{N \times n}$ (the half-open interval). There are $2$ options to choose $\epsilon, \delta$ and $N$: 
    \begin{itemize}
   \item Solve the MILP only $M=1$ times with $\epsilon < \alpha - \delta$ and $N$ samples. If
\begin{align}
N \ge \frac{1}{2(\alpha- \delta-\epsilon)^2}\ln\frac{1}{\delta} 
\end{align}
then \begin{align}
\P_{\U}(\ell^{\text{Sample}}_T(\x, \epsilon, m, N, \U) \le\ell^{\text{Uniform}}_T(\x, \alpha- \delta, m)) \ge 1- \delta. 
\end{align}
For example, if $\alpha = 0.05, \delta = 0.001, \epsilon = 0.03$ then $N = 9568$. If $\alpha = 0.05, \delta = 0.001, \epsilon = 0.048$ then $N = 3.5$ millions.  Note that $N$ is derived from Hoeffding's inequality, so it could be overly conservative. A derivation of $N$ using a tighter inequality can make it smaller. 
\item Let  $\epsilon = \frac{\lfloor N (\alpha - \delta) \rfloor}{N}$. Generate and solve $M$ MILPs  with $N$ samples and $\epsilon$, and then take the smallest result among $M$ results.
Note that to generate MILP $s$ of the $M$ MILPs, we need to resample $N$ samples of $n$ i.i.d random variables separately, denoted $\U(s)$. Each $\U(s) \in (0,1]^{N\times n}, 1 \le s \le M$ is  $N$ vectors of $n$ i.i.d uniform random variables. We solve $M$ MILPs to get $M$ values $\ell^{\text{MILP}}_T(\x, \epsilon,m, N, \U(s)),  ~1\le s \le M$, which are lower bounds of $M$ values
$\ell^{\text{Sample}}_T(\x, \epsilon,m, N, \U(s)), ~1\le s \le M$. 
Let 
\begin{align}
\underline\ell^{\text{Sample}}_T( \x, \epsilon,m, N, \U, M )= \min_{s, 1 \le s \le M} \ell^{\text{Sample}}_T(\x, \epsilon,m, N, \U(s)). 
 \end{align}
 Then if  $M \ge \log_2\left( \frac{1}{\delta}\right) $ then 

\begin{align}
\P_{\U}(\underline\ell^{\text{Sample}}_T( \x, \epsilon,m, N, \U, M ) \le\ell^{\text{Uniform}}_T(\x, \alpha-\delta, m)) \ge 1 - \delta. 
\end{align}
To make $M$ as small as possible we can first choose $\delta$, calculate $\epsilon = \frac{\lfloor N (\alpha - \delta) \rfloor}{N}$ and then reset $\delta = \alpha - \epsilon$. 
In this option $N$ can have any value, but in practice we want $N$ to be large so that the bound is tighter. 
If $\delta = 0.001$ then $M =10$. 
\end{itemize}
To simplify notations, for either option, below we use $\U \in (0,1]^{N \times n}$ to denote $N$ samples of $n$ i.i.d. uniform random variables for a single MILP. If the user chooses option $2$ they need to repeat the process of generating and solving the MILP $M$ times. 
\item $\gamma$ is the MIP gap. 
\item Let $m \in \mathbb{N}, m \ge 1$ be the size of $\v$. Recall that
\begin{align}
{\ell}^{\text{Uniform}}_T(\x, \alpha - \delta, m) \le b^{\cF}_T(\x, \alpha - \delta) \le\ell^{\text{Uniform}}_T(\x, \alpha - \delta, m) + \Delta_m.   
\end{align}
For any $m$, ${\ell}^{\text{Uniform}}_T(\x, \alpha - \delta, m)$ is guaranteed to be a lower bound of $b^{\cF}_T(\x, \alpha - \delta)$ and therefore is a valid lower bound. However if $m$ is large then the bound can be tighter. 
For example, if $m=100$ and $\theta$ is the mean then the error is within $\Delta = 1/m = 0.01$. 
\end{itemize}
\item Let $v_k:= \frac{k}{m}$.
For each $i, 1\le i \le N$: 
\begin{itemize}
    \item Sample $n$ i.i.d uniform random variables $\U^i \in (0,1]^n$ 
    \item Sort $\U^i$ in non-increasing order so that $\forall j, 1\le j \le n-1: U^i_j \ge U^i_{j+1}$. 
\end{itemize}

\item 
Solve the following $M$ MILPs within MIP gap $\gamma$ with $ T, \epsilon,  m, N, \U, M$ chosen from previous steps: 
\begin{align}
u^{\text{MILP}}_T(\x, \epsilon,m, N, \U, \gamma )= \inf \theta(H_{\v \oplus 1 , \f}) 
\end{align}
such that: 
\begin{align}
&0 \le f_1 \le \cdots \le f_m = 1 \\
&\forall i, 1 \le i \le N:  z_i \in \{0,1\} \\
&\forall i,j,k, 1 \le i \le N, 1 \le j \le n, 1\le k \le m:  w_{i,j,k} \in \{0,1\} \\
&\forall i,j, k, 1 \le i \le N, 1 \le j \le n, 1 \le k \le m: w_{i,j,k} \le  w_{i,j,k+1}\\
&\forall i,j, k, 1 \le i \le N, 1 \le j \le n, 1 \le k \le m: w_{i,j,k} \ge  w_{i,j+1,k}\\
&\sum_{i=1}^N z_i \ge \lceil N \epsilon \rceil  \\
&\forall i, j, 1 \le i \le N, 1\le j \le n: y^i_j = 1 - \frac{1}{m}\sum_{k=1}^{m} w_{i,j,k}\\
&\forall i, 1 \le i \le N:  ~T(\y^i) \ge T(\x) z_i\\
&\forall i, 1 \le i \le N:   H^{\geq}_{\v \oplus 1, \f}(\y^i)  \ge z_i \U^i  \\
&=\begin{cases}
&\forall i,j,k, 1 \le i \le N, 1 \le j \le n, 1\le k \le m: \\
&~~~~~~~~~~~~~~~~~~1 - f_{k-1} \ge  U^i_j(z_i - w_{i,j,k} ) . 
 \end{cases} 
\end{align}

This is a MILP except for the objective $\theta(H_{\v \oplus 1 , \f}) $ and the constraint $T(\y^i) \ge T(\x)$ where $T(\x)$ is a constant. If $\theta$ is the mean then the objective is $1- \frac{1}{m} \sum_{i=1}^{m-1} f_i $. If $T$ is not a linear function we need to approximate the level set $T(\y^i) \ge T(\x)$ by a set of linear constraints. 

 \item 
 Let $u$ be the lower bound of the objective that the MILP solver output. Set $\ell^{\text{MILP}}:=  u - \Delta_m$ to be the output valid lower bound. 

 In the case when $\theta$ is the mean, for implementation convenience, we maximize $1-\theta$ instead of minimizing $\theta$. Let $r$ be the upper bound of the objective $1-\theta$ that the MILP solver output. Let $u:= 1 - r$. Set $\ell^{\text{MILP}}:=  u - \Delta_m$ to be the output valid lower bound. 
\end{enumerate}

\subsection{Non-monotonic Order Function}
\label{app:non_monotonic}

We discuss the case when the order function $T$ is not monotonic, which includes student's $t$ bound.  

The discretization step (Lemma~\ref{lem:discrete_bound}) requires the monotonic condition so that the bound calculated from the discretized distribution $\ell^{\text{Discrete}}_T$ is always a lower bound of the $T$-optimal bound $b_T$. When the condition is not satisfied, as noted in the discussion after Lemma~\ref{lem:discrete_bound}, $\ell^{\text{Discrete}}_T$ might still approximate $b_T$ well when $m$ is large, but it is not a lower bound. 

The conversion to chance-constrained programs in Lemma~\ref{lem:chance_constrained_bound} requires the monotonic condition so that the constraint are equal. When the monotonic condition is not satisfied, the bound $\ell^{\text{Uniform}}_T$ in Lemma~\ref{lem:chance_constrained_bound} is a lower bound of  the discretized bound $\ell^{\text{Discrete}}_T$ in Lemma~\ref{lem:discrete_bound}, but it might be looser.

The MILP in Lemma~\ref{sec:MILP_bound} requires the monotonic condition to reduce $\y^i \in (-\infty, 1]^n$ to $\y^i \in \{v_1, \cdots, v_m\}^n$. Without the monotonic condition, the constraints on $\y^i$ can still be written as linear functions but will contain strict inequalities, which MILP solvers cannot handle. 

We leave the case of non-monotonic $T$ to future works.

\subsection{ Inverting Test Hypotheses as the Order Function}
\label{app:T_as_inverse}
In this section we consider the optimization problem where the order function is from a bound obtained by inverting test hypotheses such as \citep{waudbysmith2021estimating}. 
Such bounds output a confidence set $\Theta(\x)$ where the lower bound is $T(\x) = \inf_{\theta \in \Theta(\x)} \theta$, or equivalently if $\theta \in \Theta(\x)$ then $T(\x) \le \theta$. We assume $\Theta(\x) = \{ \theta: f_{\x}(\theta) > 0\}$. Whether the order function $T$ satisfies the monotonic condition will depend on $f_{\x}$. We discussed the case of non-monotonic function $T$ in Section~\ref{app:non_monotonic}. 

We will derive the constraint $T(\y) \ge T(\x)$ in the MILP. $T(\x)$ is already computed so it is a constant. We make a grid $g_i = \frac{i}{10000}$ over the interval $[0,1]$ to find a lower bound $m$ for $T(\y)$. Using the big M method, let binary $c_i$ be such that if $f_{\y}(g_i) > 0$ then $c_i =1$. Let binary $c'_i$ be such that if $c'_i = 1$ then $m \le g_i$. We add $c_i \le c'_i$. Then if that $f_{\y}(g_i) > 0$ then $c_i =1$, $c'_i = 1$ and then $m \le g_i$. 

\begin{align}
&\forall i, f_{\y}(g_i) \le M c_i \\
&\forall i, g_i - m \ge - M(1- c'_i)\\
&\forall i, c_i, c'_i \in \{0,1\}, c_i \le c'_i \\
&m \ge T(\x). 
\end{align}

In \citep{waudbysmith2021estimating}, the confidence set is: 
\begin{align}
\Theta(\x) = \{ \theta \in [0,1]: \Pi_{j=1}^n f_j (x_j, \lambda_j(\theta), \theta) < 1/\alpha \}
\end{align}
and the set of constraints becomes: 
\begin{align}
&\forall i, 1/\alpha -  \Pi_{j=1}^n f_j (y_j, \lambda_j(g_i), g_i) \le M c_i  \label{eq:non_linear}\\
&\forall i, g_i - m \ge - M(1- c'_i)\\
&\forall i, c_i, c'_i \in \{0,1\}, c_i \le c'_i \\
&m \ge T(\x). 
\end{align}
Eq.~\ref{eq:non_linear} makes the problem non-linear. We leave the optimization with such constraints to future works. One direction is to approximate such constraints by piece-wise linear constraints.  
\section{Experiment Details}
\label{app:exp}

To make the solver faster, we add the following constraints to the MILP. The constraints do not affect the result, but will reduce the search space and make the solver faster. 
\begin{itemize}
\item $\forall i, j, k: w_{i,j,k} \le z_i$.
\item If $\U^{i} \preceq \U^{i'}$, then $z_i \ge z_{i'}$. 
\end{itemize}

To make the solver faster, when $T(\x) = \mathbf{c}^T \x$ (in both cases where $\x$ is sorted or unsorted) we add the following constraints to the MILP. The constraints do not affect the result, but will reduce the search space and make the solver faster. 
\begin{itemize}
\item 
$\forall i: T(\y^i) - \frac{1}{m}  \min_{j: c_j > 0} c_j \le T(\x) + (\sum_{j=1}^n c_j - - \frac{1}{m}  \min_{j: c_j > 0} c_j - T(\x)) (1-z_i)$. 
\item If $c_j = 0$, then $y_j = y_{j-1}$. If $c_1 = 0$, then $y_1 = 0$. 
\end{itemize}

We also set $z_i$ to be continuous to lower the number of discrete variables. Since we have the constraint $\forall i, j, k: w_{i,j,k} \le z_i$, if $w_{i,j,k} =1$ then $z_i = 1$ so this is not a significant relaxation. 
 To generate $\U$ in $(0,1]$, we first use Numpy to generate $\U$ in $[0,1)$ and then set $\U:=1- \U$. We set the Numpy random seed to be $42$.

For the HIGHS solver we set $\{$
        "solver": "ipm",
        'run\_crossover': 'on',
        'presolve': 'on',
        'parallel': 'on',
        'threads': 8,   
        'mip\_detect\_symmetry': 'on',  
        'mip\_heuristic\_effort': 0.8
    $\}$. 

      We set the objective of the MILP to maximize $\frac{1}{m} \sum_{i=1}^{m-1} f_i$.  Let $r$ be the upper bound of the objective output by the MILP solver. Set $\ell^{\text{MILP}}:= 1 - r - \frac{1}{m}$ to be the output valid lower bound. 

      Note that although we use $\epsilon$ as the significance level of the MILP, we use the original $\alpha$ to compute $T(\x, \alpha)$ and $T(\y^i, \alpha)$ for the constraint $T(\y^i, \alpha) \ge T(\x, \alpha)$ when $T$ is Anderson's bound and Gaffke's bound. 
    \subsection{Anderson's Bound}
We use the calculation of Anderson's bound through Monte Carlo simulations from \citep{phan2021practical}, summarized below.  In one instance of Anderson's bound,  let $ \u^{And} \in [0,1]^n$ be defined as
\begin{equation}
    \label{eq:And}
     u^\text{And}_j:= \max \left \{0,j/n-\beta(n) \right \}.
\end{equation}
where  $\beta(n)$ satisfies:
\begin{align}
\P(~\forall j, 1 \le j \le n: F(X_{(j)}) \ge j/n -  \beta(n)) &\ge 
    \P_{\U} (~\forall j, 1 \le j \le n: U_{(j)} \ge j/n -  \beta(n) ) \\ &= 1 - \alpha.
\end{align}
An upper bound of the mean is then the mean of the CDF $H_{\X \oplus 1, u^\text{And}}$ which lower bound the CDF $F$: 
\begin{align}
\P(m(H_{\X \oplus 1, u^\text{And}}) \ge \mu_F) &\ge \P(~\forall j, 1 \le j \le n: F(X_{(j)}) \ge u^\text{And}) \\
&= \P(~\forall j, 1 \le j \le n: F(X_{(j)}) \ge j/n -  \beta(n)) \text{ because } F(X_{(j)}) \ge 0\\
&\ge 1 - \alpha. 
\end{align}

$\P_X( F(X) \ge y) \ge 
    \P_{U} ( U \ge y )$ is true but   $\P_X( F(X) \ge y) \le 
    \P_{U} ( U \ge y )$ is not true. Therefore we cannot calculate a lower bound of the mean from an upper bound of the CDF. Instead, we calculate a lower bound of the mean of $X$ from Anderson's upper bound of the mean of $1-X$. Let $u^\text{And}_{0}:= 0$ and assume that $X_1 \le \cdots \le X_n$. Then 
    \begin{align}
    1 - \alpha &\le \P(m(H_{(1-\X) \oplus 1, u^\text{And}}) \ge 1-\mu_F) \\
& = \P(\sum_{j=1}^n (u^\text{And}_j - u^\text{And}_{j-1})(1-X_{n+1-j}) + (1- u^\text{And}_n)\cdot 1  \ge 1-\mu_F) \\
&= \P(\sum_{j=1}^n (u^\text{And}_j - u^\text{And}_{j-1})X_{n+1-j}   \le \mu_F). 
    \end{align}

    Therefore $\sum_{j=1}^n (u^\text{And}_j - u^\text{And}_{j-1})X_{n+1-j}  $ is a valid lower bound of the mean. We called it Anderson's lower bound of the mean and also use it as the order function $T$. We use $1,000,000$ samples of vector $\U$ to calculate Anderson's coefficient using the method in \citep{phan2021practical}. 
\subsection{Gaffke's Bound}
\label{app:gaffke}

 Let $\mu(H_{\x\oplus 1,\A})$ denote the mean of the CDF $H_{\x \oplus 1,\A}$ computed by: 
\begin{itemize}
    \item Sort $\x$ and $\A$ so that $x_1 \le \cdots \le x_n \le 1$, $A_1 \le \cdots \le A_n$. 
    \item Let $x_{n+1}=1$ and $A_{n+1}=1$. Output $\mu(H_{\x,\A}):= \sum_{j=1}^{n+1} x_j (A_j - A_{j-1})$. 
\end{itemize}

We perform the following step to sample $\A$: 
\begin{enumerate}
\item Choose $N'$ to be number of samples to compute the random Gaffke's bound. The output of the MILP is a valid bound for any value of $N'\ge 1$, but the larger $N'$ is, the output will be more likely be better than Gaffke's bound. 
\item For each $i', 1\le i' \le N'$: 
\begin{itemize}
    \item Sample $n$ i.i.d uniform random variables $\A^{i'} \in [0,1]^n$. 
    \item Sort $\A^{i'}$ in non-decreasing order so that $\forall j, 1\le j \le n-1: A^{i'}_j \le A^{i'}_{j+1}$. 
\end{itemize}
 Let $\a$ denote the sample value of $\A$. We sample $\a$ just once at the beginning and use $T_{\a}$ for all samples $\x$. The bound will be valid since $T_{\a}$ is just a function. 
\end{enumerate}

 Define the function $T_{\a}(\x)$ to be: 
\begin{enumerate}
    \item Computing $m_{i'}:= \mu(H_{(1-\x) \oplus 1,\a^{i'}})$ for every $i'$. 
    \item Sort the $N'$ value computed in the previous step so that $m_1 \le m_2 \le \cdots m_{N'}$. 
    \item Let $q:= \lceil N' (1-\alpha) \rceil$. Let $\hat{Q}_{\a}(1-\alpha, 1-\x):= m_q$
\item Output $T_{\a}( \x) = 1 - \hat{Q}_{\a}(1-\alpha,1- \x)$.
\end{enumerate}

Then we perform the algorithm in Appendix~\ref{app:complete_algo} with $T_{\a}( \y^i) \ge z_i T_{\a}(\x)$. The constraint $T_{\a}(\y^i) \ge z_i T_{\a}(\x)$ is equivalent to $1 - z_i + z_i \hat{Q}_{\a}(1-\alpha, 1-\x) \ge \hat{Q}_{\a}(1-\alpha, 1-\y^i)$. When $z_i = 1$, $\hat{Q}_{\a}(1-\alpha, 1-\x) \ge \hat{Q}_{\a}(1-\alpha, 1-\y^i)$ is equivalent to $\hat{Q}_{\a}(1-\alpha, 1-\x) \ge \mu(H_{(1-\y^i) \oplus 1,\a^{i'}})$ for $
\lceil N'(1-\alpha)\rceil$ values of $i'$. Using the big M method for counting where if $c_{i,i'} = 1$ then $ \hat{Q}_{\a}(1-\alpha, 1-\x) \ge \mu(H_{(1-\y^i) \oplus 1,\a^{i'}})$: 
\begin{align}
&\forall i, ~1 \le i \le N: \\
&~~~~~\forall i', 1 \le i' \le N':  1 + c_{i, i'}(\hat{Q}_{\a}(1-\alpha, 1-\x) - 1)  \ge  \mu(H_{(1-\y^i) \oplus 1,\a^{i'}})\\
&~~~~~\sum_{i' = 1}^{N'} c_{ i,i'} \ge \lceil N'(1-\alpha) \rceil z_i  \\
&~~~~~\forall i', 1 \le i' \le N' : c_{i,i'} \le z_i \\
&~~~~~\forall i', 1 \le i' \le N' : c_{i,i'} \in \{0, 1\},
\end{align}

which is equivalent to: 
\begin{align}
&\forall i, ~1 \le i \le N: \\
&~~~~~\forall i', 1 \le i' \le N':  1 - c_{i, i'}T_{\a}(\x)  \ge  \mu(H_{(1-\y^i) \oplus 1,\a^{i'}})\\
&~~~~~\sum_{i' = 1}^{N'} c_{ i,i'} \ge \lceil N'(1-\alpha) \rceil z_i . \\
&~~~~~\forall i', 1 \le i' \le N' : c_{i,i'} \le z_i \\
&~~~~~\forall i', 1 \le i' \le N' : c_{i,i'} \in \{0, 1\}.
\end{align}
Here, when $z_i =0 $ then $c_{i,i'} =0$, so all the constraints are trivially true. 

The value of $\a$ we sampled and used in the experiment is in Table~\ref{tab:u_gaffke}. 
\begin{table}[ht]
\centering
\begin{tabular}{ccc c ccc c ccc c ccc}
\hline
1 & 2 & 3 && 1 & 2 & 3 && 1 & 2 & 3 && 1 & 2 & 3 \\
\hline
0.3745 & 0.9507 & 0.7320 & & 0.7290 & 0.7713 & 0.0740 & & 0.9083 & 0.2396 & 0.1449 & & 0.2440 & 0.9730 & 0.3931 \\
0.5987 & 0.1560 & 0.1560 & & 0.3585 & 0.1159 & 0.8631 & & 0.4895 & 0.9857 & 0.2421 & & 0.8920 & 0.6311 & 0.7948 \\
0.0581 & 0.8662 & 0.6011 & & 0.6233 & 0.3309 & 0.0636 & & 0.6721 & 0.7616 & 0.2376 & & 0.5026 & 0.5769 & 0.4925 \\
0.7081 & 0.0206 & 0.9699 & & 0.3110 & 0.3252 & 0.7296 & & 0.7282 & 0.3678 & 0.6323 & & 0.1952 & 0.7225 & 0.2808 \\
0.8324 & 0.2123 & 0.1818 & & 0.6376 & 0.8872 & 0.4722 & & 0.6335 & 0.5358 & 0.0903 & & 0.0243 & 0.6455 & 0.1771 \\
0.1834 & 0.3042 & 0.5248 & & 0.1196 & 0.7132 & 0.7608 & & 0.8353 & 0.3208 & 0.1865 & & 0.9405 & 0.9539 & 0.9149 \\
0.4319 & 0.2912 & 0.6119 & & 0.5613 & 0.7710 & 0.4938 & & 0.0408 & 0.5909 & 0.6776 & & 0.3702 & 0.0155 & 0.9283 \\
0.1395 & 0.2921 & 0.3664 & & 0.5227 & 0.4275 & 0.0254 & & 0.0166 & 0.5121 & 0.2265 & & 0.4282 & 0.9667 & 0.9636 \\
0.4561 & 0.7852 & 0.1997 & & 0.1079 & 0.0314 & 0.6364 & & 0.6452 & 0.1744 & 0.6909 & & 0.8530 & 0.2944 & 0.3851 \\
0.5142 & 0.5924 & 0.0465 & & 0.3144 & 0.5086 & 0.9076 & & 0.3867 & 0.9367 & 0.1375 & & 0.8511 & 0.3169 & 0.1695 \\
0.6075 & 0.1705 & 0.0651 & & 0.2493 & 0.4104 & 0.7556 & & 0.3411 & 0.1135 & 0.9247 & & 0.5568 & 0.9362 & 0.6960 \\
0.9489 & 0.9656 & 0.8084 & & 0.2288 & 0.0770 & 0.2898 & & 0.8773 & 0.2579 & 0.6600 & & 0.5701 & 0.0972 & 0.6150 \\
0.3046 & 0.0977 & 0.6842 & & 0.1612 & 0.9297 & 0.8081 & & 0.8172 & 0.5552 & 0.5297 & & 0.9901 & 0.1401 & 0.5183 \\
0.4402 & 0.1220 & 0.4952 & & 0.6334 & 0.8715 & 0.8037 & & 0.2419 & 0.0931 & 0.8972 & & 0.8774 & 0.7408 & 0.6970 \\
0.0344 & 0.9093 & 0.2588 & & 0.1866 & 0.8926 & 0.5393 & & 0.9004 & 0.6331 & 0.3390 & & 0.7025 & 0.3595 & 0.2936 \\
0.6625 & 0.3117 & 0.5201 & & 0.8074 & 0.8961 & 0.3180 & & 0.3492 & 0.7260 & 0.8971 & & 0.8094 & 0.8101 & 0.8671 \\
0.5467 & 0.1849 & 0.9696 & & 0.1101 & 0.2279 & 0.4271 & & 0.8871 & 0.7799 & 0.6420 & & 0.9132 & 0.5113 & 0.5015 \\
0.7751 & 0.9395 & 0.8948 & & 0.8180 & 0.8607 & 0.0070 & & 0.0841 & 0.1616 & 0.8986 & & 0.7983 & 0.6500 & 0.7020 \\
0.5979 & 0.9219 & 0.0885 & & 0.5107 & 0.4174 & 0.2221 & & 0.6064 & 0.0092 & 0.1015 & & 0.7958 & 0.8900 & 0.3380 \\
0.1960 & 0.0452 & 0.3253 & & 0.1199 & 0.3376 & 0.9429 & & 0.6635 & 0.0051 & 0.1608 & & 0.3756 & 0.0940 & 0.5783 \\
0.3887 & 0.2713 & 0.8287 & & 0.3232 & 0.5188 & 0.7030 & & 0.5487 & 0.6919 & 0.6520 & & 0.0359 & 0.4656 & 0.5426 \\
0.3568 & 0.2809 & 0.5427 & & 0.3636 & 0.9718 & 0.9624 & & 0.2243 & 0.7122 & 0.2372 & & 0.2865 & 0.5908 & 0.0305 \\
0.1409 & 0.8022 & 0.0746 & & 0.2518 & 0.4972 & 0.3009 & & 0.3254 & 0.7465 & 0.6496 & & 0.0373 & 0.8226 & 0.3602 \\
0.9869 & 0.7722 & 0.1987 & & 0.2848 & 0.0369 & 0.6096 & & 0.8492 & 0.6576 & 0.5683 & & 0.1271 & 0.5222 & 0.7700 \\
0.0055 & 0.8155 & 0.7069 & & 0.5027 & 0.0515 & 0.2786 & & 0.0937 & 0.3677 & 0.2652 & & 0.2158 & 0.6229 & 0.0853 \\
\hline
\end{tabular}
\caption{The value of $\a$ to calculate the Gaffke's bound in Section~\ref{app:gaffke}.}
  \label{tab:u_gaffke}
\end{table}

\section{Proofs}
\label{app:proof}
To simplify notations, we omit some inputs, superscripts and subscript of $b_T, \ell^{\text{Discrete}}_T, \ell^{\text{Uniform}}_T, \ell^{\text{Sample}}_T$ and $\ell^{\text{MILP}}_T$ when the context is clear. 
\subsection{Proofs of Section~\ref{sec:sub_def} and Section~\ref{sec:opt}}
We use the following lemma from~\citep{phan2021practical}:
\begin{lemma}[Probability Integral Transform]
\label{lem:cdf_vs_U}
Let $X$ be a random variable with cdf $F$ and $Y = F(X)$, known as the probability integral transform of $X$.  Then for any $0 \le y \le 1, x \in \R$, 
\begin{align}
    \P(Y \le y) &\le y.
\end{align}
\end{lemma}

\begin{proof}[Proof of Theorem~\ref{thm:correct}]
If $F \in \F$ and $F^{\ge}_{T(\X)}(T(\x))  > \alpha$, then the optimization problem is feasible with input $\x$ and $b^{\F}_T(\x, \alpha) \le \theta(F)$ by definition. 
Let $Z:= -T(\X)$ and $Y: = F_{Z}(Z)$ be the probability integral transform of $Z$. Then $Y =F^{\ge}_{T(\X)}(T(\X))$ and $\P(Y \le \alpha) \le \alpha$ by  Lemma~\ref{lem:cdf_vs_U}. 
For any $F \in \F$ we have:
\begin{align}
    \P_{F}(b^{\F}_T(\X, \alpha) \le \theta(F)) &\ge \P_{F}(F^{\ge}_{T(\X)}(T(\X)) > \alpha) \\
    &=1-  \P_{F}( Y \le \alpha)) \\
    &\ge 1 - \alpha \text{ by Lemma~\ref{lem:cdf_vs_U}}.
\end{align}

\end{proof}
\begin{proof}[Proof of Lemma~\ref{lem:helper}]
Suppose $T(\x) \le T(\y)$. Then if $G$ satisfies $\P_G(T(\X) \ge T(\y)) > \alpha $ then $G$ satisfies $\P_G(T(\X) \ge T(\x)) > \alpha $. Therefore $b^{\F}_T(\x, \alpha) \le b^{\F}_T(\y, \alpha)$. 
\end{proof}
\begin{proof}[Proof of Theorem~\ref{thm:optimal}]
Suppose there exists $\x \in \R^n$:  
$$
g(\x) > b^{\F}_T(\x, \alpha).
$$
Then the optimization problem is feasible with input $\x$ and there exists $\epsilon> 0 $ such that $g(\x) > \epsilon + b^{\F}_T(\x, \alpha)$. By Definition~\ref{def:optimal_bound}, there exists a distribution $G \in \mathcal{F}$ such that: 
\begin{itemize}
    \item $\P_G(T(\X) \ge T(\x)) > \alpha$
    \item $b^{\F}_T(\x, \alpha) \ge \theta(G)-  \epsilon$
\end{itemize}

We have $g(\x) > b^{\F}_T(\x, \alpha) + \epsilon \ge \theta(G) $. Therefore for any $\y$ such that $T(\y) \ge T(\x)$, $g(\y) \ge g(\x) > \theta(G)$. Therefore: 
\begin{align}
    \P_G (g(\X) > \theta(G)) \ge \P_G(T(\X) \ge T(\x)) > \alpha, 
\end{align}
which is equivalent to:
\begin{align}
    \P_G (g(\X) \le \theta(G))  < 1- \alpha, 
\end{align}
a contradiction. 
\end{proof}

\begin{proof}[Proof of Corollary~\ref{cor:proof_equiv}]
If $T(\x)$ is a valid lower confidence bound, then from Thm.~\ref{thm:optimal}, $\forall \x \in \R^n, T(\x) \le b^{\F}_T(\x)$. If $\forall \x\in \R^n, T(\x) \le b^{\F}_T(\x)$, since $b^{\F}_T(\x)$ is a valid lower confidence bound, $T(\x)$ is a valid lower confidence bound. 
\end{proof}
\subsection{Proofs of Section~\ref{sec:sample_mean_order}}
The result we use from \citep{luczak2016maximaltailprobabilitysums} is the following: 
\begin{lemma}[From  \citep{luczak2016maximaltailprobabilitysums}]
\label{lem:tail_bound}
Let $\mathcal{F}^{[0,\infty)} (v)$ be the set of distributions $G$ on the interval $[0, \infty)$ such that $\mu(G) \le v$. Let $c > 0$. If (i) $n=2$ and $\frac{v}{c} < x_0(2)$ or (ii) $n = 3$ and $\frac{v}{c} < x_0(3)$ or (iii) $n\ge 5$ and $\frac{v}{c}  < x_0(n)$ and $\frac{v}{c} < \frac{1}{2n-1}$:
\begin{align}
    \sup_{G \in \mathcal{F}^{[0,\infty)}(v)} \P_G(X_1 + \cdots X_n \ge c)  = f_n(\frac{v}{c}) 
\end{align}
where:
\begin{align}
    f_n(x) =  1 - (1-x)^n 
\end{align}
and $x_0(n)$ is the root of the polynomial $p_n(x) = (n x)^n + (1-x)^n -1$ in $(0, 1/n)$. When $n=2$, $x_0(n) = 2/5$ and when $n = 3, x_0(n) = \frac{-3 + \sqrt{321}} {52}$. 

A worst case distribution $G$ is a distribution on $2$ points $\{0,c\}$:
\begin{align}
    \P_G(X_i = c) &= \frac{v}{c} \\
    \P_G(X_i = 0) &= 1- \frac{v}{c}.
\end{align}
\end{lemma}

\begin{proof}[Proof of Theorem~\ref{thm:luczak_positive}]
If $\sum_{i=1}^n x_i = 0$, then $ b^{[0, \infty)}(\x) =0$. Now we assume $\sum_{i=1}^n x_i > 0$. 
From Lemma~\ref{lem:tail_bound}, given a sample $\x$, we have  if (i) $n=2$ and $\frac{v}{\sum_i x_i} < x_0(2)$ and or (ii) $n = 3$ and $\frac{v}{\sum_i x_i} < x_0(3)$ or (iii) $n\ge 5$ and $\frac{v}{\sum_i x_i}  < x_0(n)$ and $\frac{v}{\sum_i x_i} < \frac{1}{2n-1}$:

\begin{align}
    \sup_{G \in \mathcal{F}^{[0,\infty)}(v)} \P_G(X_1 + \cdots X_n \ge \sum_{i=1}^n x_i ) = f_n(\frac{v}{\sum_{i=1}^n x_i} ). \label{eq:sup_constraint}
\end{align}

If $\alpha < 1-(1- x_0(n))^n$ then let $v_{\alpha} := 
    (\sum_{i=1}^n x_i) (1 - (1-\alpha)^{1/n}).$
Then $\frac{v_{\alpha}}{\sum_{i=1}^n x_i} = 1 - (1-\alpha)^{1/n} < x_0(n)$. 

Since $x_0(n)$ cannot be easily calculated for any $n \ge 4$, we derived a simpler bound on $\alpha$. \citet{LUCZAK2014178} note that $x_0(n) > \frac{1}{n} - \frac{1}{2n^2}$. Therefore if  
$\alpha < 1 - \left( 1- \frac{1}{n} + \frac{1}{2n^2}\right)^n$ then $\alpha < 1-(1- x_0(n))^n$. 
It can be check that $1 - \left( 1- \frac{1}{n} + \frac{1}{2n^2}\right)^n > 0.3$ for all $n \in \mathbb{N}, n \ge 1$, therefore if $\alpha \le 0.3$ then $\alpha < 1-(1- x_0(n))^n$. 

 Now we show that if $n \le 10,000$ and $\alpha \le 0.3$ then $\frac{v_{\alpha}}{\sum_{i=1}^n x_i} < \frac{1}{2n-1}$. When $n \le 10,000$ we have $ 0.3 < 1 - (1- \frac{1}{2n-1})^n$. Therefore $\alpha \le 0.3 < 1 - (1- \frac{1}{2n-1})^n$, which is equivalent to $1 - (1-\alpha)^{1/n}  < \frac{1}{2n-1}$. Therefore $\frac{v_{\alpha}}{\sum_{i=1}^n x_i} < \frac{1}{2n-1}$. 
 
 So if (i) $n=2$ and $\alpha < 1-(1-2/5)^2 = 16/25$ or (ii) $n=3$ and $\alpha < 0.63 < 1-(1-\frac{-3 + \sqrt{321}}{52})^3$ or (iii) $5 \le n\le 10000$ and $\alpha \le 0.3$ we apply Eq.~\ref{eq:sup_constraint} with $v = v_{\alpha}$ and have:
\begin{align}
    \sup_{G \in \mathcal{F}^{[0, \infty)}(v_{\alpha})} \P_G(\frac{X_1 + \cdots X_n}{n} \ge \frac{\sum_{i=1}^n x_i}{n}) &= f_n(\frac{v_{\alpha}}{\sum_{i=1}^n x_i} ) \\&= \alpha. \label{eq:equal_constraint}
\end{align}
Recall that:  
$$
b^{[0, \infty)}(\x)=  \inf_{G \in \mathcal{F}^{[0,\infty)}} \; \mu(G)  
$$
subject to 
\begin{align}
     \P_G(\frac{X_1 + \cdots X_n}{n} \ge \frac{\sum_{i=1}^n x_i}{n})  >\alpha. 
\end{align}

From Eq.~\ref{eq:equal_constraint}, 
if $\P_G(T(\X) \ge T(\x)) >\alpha$ then $\mu(G) > v_{\alpha}$. Therefore $b^{[0, \infty)}(\x) \ge v_{\alpha}$. 

It remains to show that $b^{[0, \infty)}(\x) = v_{\alpha}$. For any $\epsilon > 0$, let $\delta = \frac{\epsilon}{\sum_{i=1}^n x_i}$ and construct the following distribution $G$: 
\begin{align}
   &\P_G(X = \sum_{i=1}^n x_i) = \frac{v_{\alpha}}{\sum_{i=1}^n x_i} + \delta \\
    &\P_G(X = 0) = 1 - \frac{v_{\alpha}}{\sum_{i=1}^n x_i} - \delta. 
\end{align}
Then: 
\begin{align}
\mu(G) &= (\frac{v_{\alpha}}{\sum_{i=1}^n x_i} + \delta) (\sum_{i=1}^n x_i) \\
&= v_{\alpha}  + \epsilon. \\
    \P_G(\sum_{i=1}^n X_i \ge \sum_{i=1}^n x_i) &= 1 - (1 - \frac{v_{\alpha}}{\sum_{i=1}^n x_i} - \delta)^n \\
    &> 1 - (1 - \frac{v_{\alpha}}{\sum_{i=1}^n x_i})^n \\
    &=\alpha. 
\end{align}
Since we can always construct $G$ on $[0, \infty)$ such that $ \P_G(\sum_{i=1}^n X_i \ge \sum_{i=1}^n x_i) > \alpha$ and $\mu(G)$ is $\epsilon$ close to $v_{\alpha} $, we have $b^{[0, \infty)}(\x) = v_{\alpha}$.

We now consider the case when $\F$ is the set of all distributions on $[0,1]$. 
If $c \le 1$, the worst case distribution on $\{0, c\}$ satisfies the condition that the distribution has support on $[0,1]$. Therefore it can only give us lower confidence bound when $\F$ is the set of all distributions on $[0,1]$ for samples with $\sum_{i=1}^n x_i  \le 1$. 
\end{proof}

\subsection{Proofs of Section~\ref{sec:MILP}}

\subsubsection{Proofs of Section~\ref{sec:discrete_bound}}

\begin{proof}[Proof of Lemma~\ref{lem:discrete_bound}]
Let $\cC:= \{ F \in \F: C(F)\}$ be the constraint on the CDF $F$ to find $b^{\F}_T$. Then:
\begin{align}
b^{\F}_T(\x) = \inf_{F \in \cC} \theta(F). 
\end{align}

Let
$\cV:= \{ 0 \le f_1 \le \cdots \le f_m = 1: C(H_{\v \oplus 1, \f})\}$ be the constraint on the CDF $H_{\v \oplus 1, \f}$ to find $\ell_T$. Then:
\begin{align}
\ell^{Discrete}_T(\x, m) = \inf_{\f \in \cV} \theta(H_{\v \oplus 1 , \f}) - \Delta_m. 
\end{align}
If $\cC$ is not empty, then there exists $F \in \F$ such that $\P_F(T(\X) \ge T(\x) > \alpha$. Since $H_{\v \oplus 1, F(\v)}$ first order stochastically dominates $F$, $C(H_{\v \oplus 1, F(\v)})$ is true. Therefore $\cV$ is not empty. If $\cC$ is empty, then $\cV$ will also be empty. So the optimization problem to find $b_T$ is feasible if and only if the optimization problem to find $u_T$ is feasible.

If the problem to find $b_T$ is infeasible, then all bounds $b_T, u_T$ and $\ell_T$ are $\infty$ adn the conclusion is trivial. So we now assume that the problem is feasible. 

\begin{enumerate}
\item 
To compute $u_T(\x, m)$ we are searching over a subset of $\F$. Therefore $u_T(\x, m) \ge b^{\F}_T(\x)$. 
\item 
We will now show that $b^{\F}_T(\x) \ge \ell^{Discrete}_T(\x, m)$.

For any CDF $F \in \cC$, we will show that $ F(\v) \in \cV$ and the CDF $H_{\v \oplus 1, F(\v)}$ satisfies: 
\begin{align}
\theta(H_{\v \oplus 1, F(\v)}) - \Delta_m \le \theta(F).  
\end{align}
\begin{itemize}
\item Since the support is lower bounded by $0$, the CDF $F$ first order stochastically dominates $H_{\v \oplus 0, F(\v)}$. Therefore $\theta(F) \ge \theta(H_{\v \oplus 0, F(\v)})$. We have: 
\begin{align}
    \theta(H_{\v \oplus 1, F(\v)}) - \theta(F)&\le   \theta(H_{\v \oplus 1, F(\v)})  - \theta(H_{\v \oplus 0, F(\v)}) \\
    &\le \Delta_m. 
\end{align}
When $\theta$ is the mean we have: 
\begin{align}
\Delta_m &= \sup_{\f} (\theta(H_{\v \oplus 1, F(\v)})  - \theta(H_{\v \oplus 0, F(\v)})) \\
 &= \sup_{\f} \left(1- \frac{1}{m}\sum_{i=1}^{m-1} f_i\right)- \left(1- \frac{1}{m}\sum_{i=1}^{m} f_i \right)\\
  &= \frac{1}{m}. 
  \end{align}
\item Now we will show that if $F \in \cC$ then $F(\v) \in \cV$. Since the support  is upper bounded by $1$, the CDF $H_{\v \oplus 1, F(\v)}$ first-order stochastically dominates the CDF $F$. Then $\P_F(T(\X) \ge T(\x)) > \alpha$ implies $\P_{H_{\v \oplus 1, F(\v)}}(T(\X) \ge T(\x)) > \alpha$.

Then $C(F)$ implies $C(H_{\v \oplus 1, F(\v)})$. Therefore if $F \in \cC$ then $F(\v) \in \cV$. 
\end{itemize}

Therefore $\ell^{Discrete}_T(\x, m) \le b^{\F}_T(\x)$. 
\item 
It remains to show that $u^{\text{Discrete}}_T(\x, m) - b^{\F}_T(\x) \le \Delta_m$ and $
b^{\F}_T(\x) - \ell^{\text{Discrete}}_T(\x, m) \le \Delta_m$. The gap between the upper bound $u^{\text{Discrete}}_T(\x, m)$ and the lower bound $\ell^{\text{Discrete}}_T(\x, m)$ can be bounded by $\Delta_m$. Therefore the gap between $b^{\F}_T(\x)$ and the upper bound (or the lower bound) is at most $\Delta_m$. 
\end{enumerate}
\end{proof}

\subsubsection{Proofs of Section~\ref{sec:chance_constraint}}
Now we will prove Lemma~\ref{lem:constraint_equal} and Lemma~\ref{lem:chance_constrained_bound}. 
Let $\S(\x)  = \{\y \in \R^n: T(\y) \ge T(\x)\}$. Let $X^-:= -X$, $\y^-:= -\y$ and $\S^-(\x)  = \{\y^- \in \R^n: T(\y) \ge T(\x)\}$.

\begin{lemma}\label{lem:constraint_equal_helper}
 Let $T: \R^n \rightarrow \R$ be a statistic of a sample of size $n$. Let $\U$ be $n$ i.i.d samples from the closed interval $[0,1]$. Then for any distribution $F$ and any $\x \in \R^n$:
\begin{align}
\P_{F}(T(\X) \ge T(\x)) \le\P_{\U} (\exists \y \in \R^n, T(\y) \ge T(\x):  \U \leq F^{\geq}(\y)) .
\end{align}

If $T$ is monotonic then for any distribution $F$ and any $\x \in \R^n$:
\begin{align}
\P_{F}(T(\X) \ge T(\x)) =\P_{\U} (\exists \y \in \R^n, T(\y) \ge T(\x):  \U \leq F^{\geq}(\y)) .
               \end{align}
\end{lemma}
\begin{proof}[Proof of Lemma~\ref{lem:constraint_equal_helper}]

The proof of the first statement follows closely the proof in ~\citep{phan2021practical} with minor modifications.  Let $G$ denote the CDF  of $X^-$. Then $G(-y) = \P(- X \le -y) = \P(X \ge y) = F^{\ge}(y)$.

\begin{align}
     \P_{\X} (T(\X) \ge T(\x)) 
    & = \P_{\X}(\exists \y \in \S(\x): \X  = \y) 
    \\ 
    & \le \P_{\X} (\exists {\y \in \S(\x)}: \X^- \preceq \y^-) \text{ because $\X = \y$ implies $\X^- \preceq \y^-$ } \\
    &\le \P_{\X} (\exists {\y \in \S(\x)}: G(\X^-) \preceq G(\y^-)) \\
     &\le \P_{\X} (\exists {\y \in \S(\x)}: G(\X^-) \preceq F^{\ge}(\y)) \\
    &\le \P_{\U} (\exists {\y \in\S(\x)}: \U \preceq F^{\ge}(\y)  ) \text{ by Lemma~\ref{lem:cdf_vs_U}. }
    \end{align}

Now we prove the second statement by showing that if $T$ is monotonic then $\P_{\X} (T(\X) \ge T(\x))  \ge \P_{\U} (\exists {\y \in\S(\x)}: \U \preceq F^{\ge}(\y)  )$\footnote{Gemini contributed to the derivation of the following part of the proof.}. 
The proof uses some ideas from the proof in \citep{phan2021practical}. 
Let $G^{-1}(u) := \inf \{x \in \R:  G(x) \ge u\}$ be the inverse of the CDF $G$. 
From \citep{Embrechts2013}, if $U$ is an uniform random variable in $[0,1]$, $G^{-1}(U)$ has CDF $G$. 

  \begin{align}
    \P_{F}(T(\X) \ge T(\x)) 
    =&\P_{F}(\X \in \R^n, T(\X) \ge T(\x)) \\ 
    =&\P_{G}(\X^- \in \R^n, T(\X) \ge T(\x)) \\ 
    =& \P_{G}(\X^- \in \S^-(\x)) \\ 
    = &\P_{\U} (G^{-1}(\U) \in \S^-(\x) ) .
    \end{align}

    We will show that if $\exists \y^- \in \S^-(\x): \u \leq G(\y^-)$ then $G^{-1}(\u) \in \S^-(\x)$.
Suppose $\exists \y^- \in \S^-(\x): \u \le G(\y^-)$, denoted ${\y^-}^*$ and let $\y^*:= - {\y^-}^*$. 

By definition $G^{-1}(\u) = \inf \{\y^- \in \R:  \u \leq G(\y^-)\}$ so $G^{-1}(\u) \preceq {\y^-}^*$. Then $\y^* \preceq - G^{-1}(\u)$, and $T(\y^*) \le T(- G^{-1}(\u))$ because $T$ is monotonic. 

Since ${\y^-}^* \in \S^-(\x)$, we have $T(\y^*) \ge T(\x)$. Therefore  $T(- G^{-1}(\u)) \ge T(\y^*) \ge T(\x)$, and $G^{-1}(\u) \in \S^-(\x)$.

Therefore:
\begin{align}
\P_{\X}(T(\X) \ge T(\x)) =& \P_{\U} (G^{-1}(\U) \in \S^-(\x) ) \\
\ge &
\P_{\U} (\exists \y^- \in \S^-(\x):  \U \preceq G(\y^-)) \\
=&\P_{\U} (\exists \y \in \S(\x):  \U \preceq F^{\geq}(\y)) .
\end{align}

\end{proof}

\begin{proof}[Proof of Lemma~\ref{lem:constraint_equal}]
When $\U$ is sampled from $(0,1]$, $\U \succ 0$. Since $P_F(X\le 1) = 1$, we have if $P(X \ge y) > 0$ then $y \le 1$. 
Let $P_{\U[0,1]}$ denote the probability when $\U$ is sampled from $[0,1]$ and $P_{\U(0,1]}$  denote the probability when $\U$ is sampled from $(0,1]$. 

Therefore from Lemma~\ref{lem:constraint_equal_helper}: 
\begin{align}
\P_{\X}(T(\X) \ge T(\x)) 
\le &\P_{\U[0,1]} (\exists \y \in \R^n:  T(\y) \ge T(\x), \U \preceq F^{\geq}(\y)) \\
=&\P_{\U(0,1]} (\exists \y \in \R^n:  T(\y) \ge T(\x),   \U \preceq F^{\geq}(\y)) \\
=&\P_{\U(0,1]} (\exists \y \in (-\infty, 1]^n:  T(\y) \ge T(\x),   \U \preceq F^{\geq}(\y)) .
\end{align}

If $T$ is monotonic then 
\begin{align}
\P_{\X}(T(\X) \ge T(\x)) 
=&\P_{\U[0,1]} (\exists \y \in \R^n:  T(\y) \ge T(\x), \U \preceq F^{\geq}(\y)) \\
=&\P_{\U(0,1]} (\exists \y \in \R^n:  T(\y) \ge T(\x),   \U \preceq F^{\geq}(\y)) \\
=&\P_{\U(0,1]} (\exists \y \in (-\infty, 1]^n:  T(\y) \ge T(\x),   \U \preceq F^{\geq}(\y)) 
\end{align}
because $F^{\geq}(\y) \succeq \U \succ 0$ implies $\y \preceq 1$. 
\end{proof}

\begin{proof}[Proof of Lemma.~\ref{lem:chance_constrained_bound}]
The proof follows directly from Lemma~\ref{lem:constraint_equal}. 
\end{proof}

\subsubsection{Proofs of Section~\ref{sec:sample_approximation}}
\label{app:sample_approximation}
We approximate the chance-constrained program $\ell^{\text{Uniform}}_T(\x, m)$ by sample approximation \citep{luedtke2007sample}. We analyze how the chance-constrained program can be lower bounded or upper bounded by the sample approximation. 

\begin{itemize}
    \item \textbf{Lower bound.} We show that the sample approximation $\ell^{\text{Sample}}_T(\x, \epsilon)$ is a lower bound of ${\ell}^\text{Uniform}_T(\x, \alpha-\delta)$ with high probability. We prove the following results with $\alpha$ and $\epsilon$ and then substitute $\alpha$ with $\alpha - \delta$ and $\epsilon$ with $\epsilon$.  Recall that $C^{\text{Uniform}}(H^{\geq}_{\v \oplus 1, \f}(\y), \alpha)$ denote the constraint: 
\begin{align}
\P_{\U(0,1]} (\exists \y \in (-\infty, 1]^n: T(\y) \ge T(\x),  \U \leq H^{\geq}_{\v \oplus 1, \f}(\y))  > \alpha
\end{align}
and $C^{\text{Sample}}(H^{\geq}_{\v \oplus 1, \f}, \epsilon, N, \U )$ define the constraint:  
    \begin{align}
  \frac{1}{N}\sum_{i=1}^N \I\left[  \exists \y^i \in  (-\infty, 1]: T(\y^i) \ge T(\x),      \U^i \leq H^{\geq}_{\v \oplus 1, \f}(\y^i)  \right] \ge  \epsilon . 
\end{align}
    
    Following \citep{luedtke2007sample}, we define the event $\text{Fail}(\f, \u):= \{\exists \y \in (-\infty, 1]^n: T(\y) \ge T(\x),  \u \leq H^{\geq}_{\v \oplus 1, \f}(\y) \}$. 
Let
\begin{align}
\cV_{\alpha} &:= \{\f \in \R^n: \P_{\U}(\text{Fail}(\f, \U)) > \alpha\}\\
\hat \cV_{\epsilon} &:= \{\f \in \R^n: \frac{1}{N} \sum_{i=1}^N \mathbb{I}(\text{Fail}(\f, \U^i)) \ge \epsilon\}. 
\end{align}
Let $\f^*_{\alpha}$ be the optimal solution that results in ${\ell}^{\text{Uniform}}_T(\x, \alpha)$ if the problem is feasible. Then $\P(\text{Fail}(\f^*_{\alpha}, \U^i)) >  \alpha$. 
If $\f^*_{\alpha} \in \hat \cV_{\epsilon}$ then ${\ell}^{\text{Uniform}}_T(\x, \alpha) \ge \ell^{\text{Sample}}_T(\x, \epsilon)$. We prove Lemma~\ref{lem:sample_approximation_bound} by proving Lemma~\ref{lem:sample_approximation_bound1}  and Lemma~\ref{lem:sample_approximation_bound2} 

\begin{lemma}[Minor adjustment from Theorem 3 in \citep{luedtke2007sample}]
\label{lem:sample_approximation_bound1}
For $N$ samples of $\U$ and probability $\epsilon <  \alpha$: 
\begin{align}
\P_{\U}(\ell^{\text{Sample}}_T(\x, \epsilon) \le\ell^{\text{Uniform}}_T(\x, \alpha)) \ge 1 - \exp(-2N(\alpha - \epsilon)^2).
\end{align}
Therefore if
\begin{align}
N \ge \frac{1}{2(\alpha-\epsilon)^2}\ln\frac{1}{\delta}
\end{align}
then
\begin{align}
\P_{\U}(\ell^{\text{Sample}}_T(\x, \epsilon) \le\ell^{\text{Uniform}}_T(\x, \alpha)) \ge  1 - \delta. 
\end{align}
\end{lemma}
\begin{proof}[Proof of Lemma~\ref{lem:sample_approximation_bound1}]
If $\ell^{\text{Uniform}}$ is infeasible, the claim is trivially true. So we assume that $\ell^{\text{Uniform}}$ is feasible and $\f^*_{\alpha}$ exists.

Theorem $3$ in \citep{luedtke2007sample} considers the case where the chance-constrained program has the form $\P(g(\f, \U) \le 0) \ge 1 - \alpha$ where $g(\f, \U)$ could be a vector and $\le$ is applied component-wise. Our program has the form:
\begin{align}
\P_{\U(0,1]} (\exists \y \in (-\infty, 1]^n: T(\y) \ge T(\x),  \U \leq H^{\geq}_{\v \oplus 1, \f}(\y))  > \alpha. 
\end{align}

The proof of Theorem $3$ still work for this case. We reproduce it with minor adjustment for our case  for completeness. 

 Let $Z_i := 1 - \mathbb{I}( \text{Fail}(\f^*_{\alpha}, \U^i))$ be indicator variables of success. Then $\E[Z_i] = 1- \P(\text{Fail}(\f^*_{\alpha}, \U^i)) < 1-  \alpha$. We have: 
\begin{align}
\P(\ell^{\text{Uniform}}_T(\x, \alpha) <  \ell^{\text{Sample}}_T(\x, \epsilon)) &\le \P ( \f^*_{\alpha} \notin \hat \cV_{\epsilon}) \\
&= \P \left( \frac{1}{N} \sum_i Z_i > 1-  \epsilon \right) \\
&\le \P \left( \frac{1}{N} \sum_i Z_i \ge 1-  \epsilon \right) \\
&\le \P \left( \frac{1}{N} \sum_i Z_i  - \E[Z_i] \ge 1- \epsilon - (1-\alpha) \right) \\
&\le \P \left( \frac{1}{N} \sum_i Z_i  - \E[Z_i] \ge \alpha- \epsilon \right)  \\
&\le \exp\left\{ -2 N(\alpha - \epsilon)^2\right\} \text{ by Hoeffding's inequality.}
\end{align}
 
Therefore: 
\begin{align}
\P(\ell^{\text{Uniform}}_T(\x, \alpha) \ge  \ell^{\text{Sample}}_T(\x, \epsilon)) \ge 1 - \exp\left\{ -2 N(\alpha - \epsilon)^2\right\}.
\end{align}
\end{proof}

We now consider the case where the value of the above $N$ is too large for the MILP. We use another result from Theorem $4$ in \citep{luedtke2007sample} with a modification\footnote{Gemini contributed to the following derivation.}.
We take $\epsilon = \frac{\lfloor N \alpha\rfloor}{N}$, so that $\epsilon  N$ is an integer. We solve the MILP $M$ times each with $\epsilon$ and $N$ samples (note that for each MILP $s$ of the $M$ MILPs, we need to resample $N$ samples of $n$ i.i.d random variables separately, denoted $\U(s)$)  to get $M$ values $\ell^{\text{Sample}}_T(\x, \epsilon,m, N, \U(1)),  \cdots , \ell^{\text{Sample}}_T(\x, \epsilon, m, N, \U(M))$. Each $\U(s) \in (0,1]^{N\times n}, 1 \le s \le M$ is  $N$ vectors of $n$ i.i.d uniform random variables and $\U(s)^{i} \in (0,1]^{n}, 1 \le s \le M, 1 \le i \le N$ is a vector of $n$ i.i.d. uniform random variables. Then we take the smallest value (denoted $\underline \ell^{\text{Sample}}_T( \x, \epsilon,m, N, \U, M )$).  Define
\begin{align}
\rho(N, \alpha, \epsilon):= \sum_{i = 0}^{\lfloor N (1-\epsilon) \rfloor}  {\binom{N}{i}}  (1-\alpha)^i \alpha^{N-i}
\end{align}
to be the probability of having at most $\lfloor N (1-\epsilon)\rfloor$ successes where the probability of success in $1$ trial is $1-\alpha$. 
\begin{lemma}[Minor modifications from Theorem $4$ in  \citep{luedtke2007sample}]
\label{lem:helper_analysis}
Let $\alpha, \epsilon, \delta \in (0,1)$ and $N, M \ge 1$ be positive integers.
If
\begin{align}
(1 - \rho(N, \alpha, \epsilon))^M \le \delta
\end{align}
then 
\begin{align}
\P_{\U}(\underline \ell^{\text{Sample}}_T(\x, \epsilon, m, N, \U, M) \le\ell^{\text{Uniform}}_T(\x, \alpha)) \ge 1 - \delta. 
\end{align}
\end{lemma}
\begin{proof}
If $\ell^{\text{Uniform}}$ is infeasible, the claim is trivially true. So we assume that $\ell^{\text{Uniform}}$ is feasible and $\f^*_{\alpha}$ exists. 

First we will prove that for $1 \le s \le M$, $\P_{\U(s)}({\ell^{\text{Sample}}_T(\x, \epsilon, N, \U(s))} \le\ell^{\text{Uniform}}_T(\x, \alpha)) \ge \rho(N, \alpha, \epsilon)$ (Lemma $1$ in \citep{luedtke2007sample}). $\rho(N, \alpha, \epsilon)$ is the probability of having at most $\lfloor N (1-\epsilon)\rfloor$ successes where the probability of success in $1$ trial is $1-\alpha$. Let
\begin{align}
\hat \cV_{\epsilon, s} &:= \{\f \in \R^n: \frac{1}{N} \sum_{i=1}^N \mathbb{I}(\text{Fail}(\f, \U(s)^i)) \ge \epsilon\}.
\end{align}
If $\f^*_{\alpha} \in \hat \cV_{\epsilon, s}$ then $\ell^{\text{Sample}}_T(\x, \epsilon, N, \U(s)) \le\ell^{\text{Uniform}}_T(\x, \alpha)$. Therefore
\begin{align}
\P_{\U(s)}({\ell^{\text{Sample}}_T(\x, \epsilon, N, \U(s))} \le\ell^{\text{Uniform}}_T(\x, \alpha)) \ge  
\P_{\U(s)}(\f^*_{\alpha} \in \hat \cV_{\epsilon, s}).
\end{align}
$\f^*_{\alpha} \in \hat \cV_{\epsilon, s}$ if there are at most $\lfloor N(1-\epsilon)\rfloor$ successes when the probability of success in $1$ trial is $1 - \P(\text{Fail}(\f^*_{\alpha}, \U(s)^i)) < 1 - \alpha$. Then for all $s, 1 \le s \le M$: 
\begin{align}
&\P_{\U(s)}({\ell^{\text{Sample}}_T(\x, \epsilon, N, \U(s))} \le\ell^{\text{Uniform}}_T(\x, \alpha)) \\
&\ge
\P_{\U(s)}(\f^*_{\alpha} \in \hat \cV_{\epsilon, s}) \\&\ge\rho(N, \P(\text{Fail}(\f^*_{\alpha}, \U(s)^1)), \epsilon) \\ &\ge  \rho(N, \alpha, \epsilon) \text{ because }1 - \P(\text{Fail}(\f^*_{\alpha}, \U(s)^1)) < 1 - \alpha.
\end{align}

Now we have: 
\begin{align}
&\P_{\U}(\underline\ell^{\text{Sample}}_T(\x, \epsilon, m, N, \U, M) \le\ell^{\text{Uniform}}_T(\x, \alpha)) \\
&= 1 - \P_{\U}(\underline\ell^{\text{Sample}}_T(\x, \epsilon, m, N, \U, M) > \ell^{\text{Uniform}}_T(\x, \alpha)) \\
&= 1 - \P_{\U}({\ell^{\text{Sample}}_T(\x, \epsilon, N, \U(1))} > \ell^{\text{Uniform}}_T(\x, \alpha), \cdots, {\ell^{\text{Sample}}_T(\x, \epsilon, N, \U(M))} >\ell^{\text{Uniform}}_T(\x, \alpha)) \\
&= 1 - \Pi_{s=1}^M \P_{\U(s)}({\ell^{\text{Sample}}_T(\x, \epsilon, N, \U(s))} > \ell^{\text{Uniform}}_T(\x, \alpha)) \\
&= 1 - \Pi_{s=1}^M (1- \P_{\U(s)}({\ell^{\text{Sample}}_T(\x, \epsilon, N, \U(s))} \le \ell^{\text{Uniform}}_T(\x, \alpha)) \\
&\ge 1 - (1-\rho(N,\alpha, \epsilon))^M\\
&\ge 1 - \delta. 
\end{align}
\end{proof}
\begin{lemma}
\label{lem:sample_approximation_bound2}
Let $\epsilon, \delta \in (0,1)$ and $N, M \ge 1$ be positive integers such that $M \ge \log_2\frac{1}{\delta}$ and $N \epsilon$ is an integer. Then if $\epsilon \le \alpha$: 
\begin{align}
\P_{\U}(\underline\ell^{\text{Sample}}_T(\x, \epsilon, m, N, \U, M) \le\ell^{\text{Uniform}}_T(\x, \alpha)) \ge 1 - \delta. 
\end{align}
\end{lemma}
\begin{proof}
When $(1-\epsilon) N$ is an integer, the mean and median of binomial random variables are the same \citep{Lord_2010}. Therefore $\rho(N, \epsilon, \epsilon) \ge 1/2$. 
If  $M \ge \log_2\left( \frac{1}{\delta}\right)$ then $(1 - \rho(N, \epsilon, \epsilon))^M \le \delta$. So by Lemma~\ref{lem:helper_analysis}: 
\begin{align}
\P_{\U}(\underline\ell^{\text{Sample}}_T(\x, \epsilon, m, N, \U, M) \le\ell^{\text{Uniform}}_T(\x, \epsilon)) \ge 1 - \delta. 
\end{align}
Since $\epsilon \le \alpha$, ${\ell}^{\text{Uniform}}_T(\x, \epsilon) \le\ell^{\text{Uniform}}_T(\x, \alpha)$. Therefore: 
\begin{align}
\P_{\U}(\underline\ell^{\text{Sample}}_T(\x, \epsilon, m, N, \U, M) \le\ell^{\text{Uniform}}_T(\x, \alpha)) \ge 1 - \delta. 
\end{align}
\end{proof}

\item \textbf{Upper bound.} We leave finding a sample-approximation ${u}^{\text{Sample}}_T(\x, \epsilon, m,  N, \U )$ which is an upper bound of the chance-constrained program $u^{\text{Uniform}}_T(\x, \alpha, m)$ to future works. 
\end{itemize}

\subsubsection{Proofs of Section~\ref{sec:MILP_bound}}
\label{app:MILP}
We consider the case when we use the MILP to maximize $1-\theta$. Let $r$ be the upper bound of the objective $1-\theta$ that the MILP solver output. Then we calculate $u^{\text{MILP}} = 1- r$, and $ \ell^{\text{MILP}} = u^{\text{MILP}} - \Delta_m$. 

\begin{proof}[Proof of Lemma~\ref{lem:milp_bound}]
The proof follows directly from the conversion to the MILP. Given MIP gap $\delta > 0$, solvers  can find a upper bound $ r$ of a maximization MILP program with true objective $r^*:= 1- \theta^*$ such that: 
\begin{align}
\frac{r}{1+\delta} \le r^* \le r. 
\end{align}
Therefore: 
\begin{align}
1 - \frac{r}{1+\delta} \ge \theta^* \ge 1- r. 
\end{align}
\end{proof}

\subsubsection{Proofs of Section~\ref{sec:analysis}}

\begin{proof}[Proof of Theorem~\ref{thm:analysis_extend}]
The first $4$ results are from previous sections. Let 
\begin{align}
s^*  = \arg\min_{s, 1 \le s \le M} \ell^{\text{Sample}}_T(\X, \epsilon, m,  N, \U(s),  \gamma). 
\end{align}
We will omit some inputs of the bound $\ell $ and $b$ when the context is clear to simplify notations. We have: 
\begin{align}
&\P_{\U, \X} (\underline \ell^{\text{MILP}}_T(\X, \epsilon, m,  N, \U, M, \gamma) \le \theta)\\ 
&\ge \P_{\U, \X} (\ell^{\text{MILP}}_T(\X, \epsilon, m,  N, \U(s^*),  \gamma) \le \theta)\\ 
&\ge \P_{\U, \X} (\underline \ell^{\text{Sample}}_T(\X, \epsilon, m,  N, \U, M, \gamma) \le \theta)  \\
&\ge \P_{\U, \X}(\underline\ell^{\text{Sample}}_T(\X, \epsilon) \le\ell^{\text{Uniform}}_T(\X, \alpha - \delta) \text{ AND } \ell^{\text{Uniform}}_T(\X, \alpha - \delta) \le \theta) \\
&= 1 - \P_{\U, \X}(\underline\ell^{\text{Sample}}_T(\X, \epsilon) > \ell^{\text{Uniform}}_T(\X, \alpha - \delta) \text{ OR } \ell^{\text{Uniform}}_T(\X, \alpha - \delta) >  \theta) \\
&\ge 1 - \P_{\U, \X}(\underline\ell^{\text{Sample}}_T(\X, \epsilon) >\ell^{\text{Uniform}}_T(\X, \alpha - \delta)) - \P_{\X}({\ell}^{\text{Uniform}}_T(\X, \alpha - \delta) > \theta) \text{ by union bound}\\
&\ge 1  -\delta - (\alpha - \delta)\\
&= 1 - \alpha .
\end{align}
\end{proof}

\section{Connections to Pivoting the CDF and Inverting Hypothesis Tests}
\label{app:connection}
\subsection{Connection to Pivoting the CDF}
\label{app:pivoting_the_CDF}
In this section we will show that our bound is a generalization of the method of pivoting the CDF. The method of pivoting the CDF \citep{CaseBerg:01} is defined as follow. 
\begin{lemma}[Pivoting a continuous (discrete) CDF \citep{CaseBerg:01}, with minor modifications]
\label{lem:pivot_cdf}
Let $T$ be a statistic with continuous (discrete) cdf $F_T(t|\theta)$. Let $0 < \alpha < 1$. Suppose that for each $t \in \mathcal{T}$, the sample space of $T$, the function $\theta_L(t)$ can be defined as follows: 
\begin{itemize}
    \item If $F_{-T}( -t |\theta)$ is an increasing function of $\theta$ for each $t$, $\theta_L(t)$ is such that:
    \begin{align}
       \P(T \ge t | \theta_L(t)) = \alpha. 
    \end{align}
    \item If $F_T(t|\theta)$ is an increasing function of $\theta$ for each $t$, $\theta_L(t)$ is such that:
    \begin{align}
        \P(T \le t | \theta_L(t)) = \alpha. 
    \end{align}
\end{itemize}
Then $\theta_L(T)$ is a $(1-\alpha)$ lower confidence bound for $\theta$. 
\end{lemma}

We show below that when certain conditions are satisfied, the bound defined in Definition~\ref{def:optimal_bound} is equal to the bound defined by pivoting the CDF in Lemma~\ref{lem:pivot_cdf}. 
\begin{theorem}

Let $b^{\mathcal{F}}_T$ be the bound defined in Definition~\ref{def:optimal_bound} and $\theta_L$ be defined in Lemma~\ref{lem:pivot_cdf}. 

Let $\mathcal{F}$ be the set of  distributions $F$ parameterized by $\theta$ such that $F_{-T}(-t|\theta)$ is an increasing function of $\theta$ for each $t$. Then for any sample $\x$:
\begin{align}
    b^{\mathcal{F}}_{T}(\x) = \theta_L(T(\x)). 
\end{align}

Let $\mathcal{F}$ be the set of distributions $F$ parameterized by $\theta$ such that $F_T(t|\theta)$ is an increasing function of $\theta$ for each $t$. Then for any sample $\x$:
\begin{align}
    b^{\mathcal{F}}_{-T}(\x) = \theta_L(T(\x)). 
\end{align}
\end{theorem}
\begin{proof}

We have $\P_{\theta} (T(\X) \ge T(\x)) = \P_{\theta} ( - T(\X) \le -T(\x) ) = F_{-T}(-t|\theta)$ is an increasing  function of $\theta$ for each $t$. By Definition~\ref{def:optimal_bound}, $b^{\mathcal{F}}_{T}(\x)$ is the solution of: 
\begin{align}
    \inf_{\theta} \P_{\theta} (T(\X) \ge T(\x)) \gneq \alpha, 
\end{align}
which is equal to $\theta_L(T(\x))$ because $\P_{\theta} (T(\X) \ge T(\x))$ is an increasing function of $\theta$. 

Suppose $\P_{\theta} (T(\X) \le T(\x))$ is an increasing  function of $\theta$ for each $t$. By Definition~\ref{def:optimal_bound}, $b^{\mathcal{F}}_{-T}(\x)$ is the solution of: 
\begin{align}
    \inf_{\theta} \P_{\theta} (- T(\X) \ge - T(\x)) \gneq \alpha, 
\end{align}
which is equal to $\theta_L(T(\x))$  when $\P_{\theta} (T(\X) \le  T(\x))$ is an increasing function of $\theta$. 
\end{proof}

\subsection{Connection to Inverting Test Hypotheses}
\label{app:generalization}

In this section we show that Definition~\ref{def:optimal_bound}  is a minor modification to the method of inverting test hypothesis. Inverting test hypotheses constructs a high-confidence set of the parameter $\theta$, but our method constructs a high-confidence set of the CDF, and then find the lower bound of $\theta$ in the high-confidence set of the CDF.  
We generalize the Definition~\ref{def:optimal_bound} to the case when there is a set $\cC(\X) \subseteq \mathcal{F}$ which is a high confidence set for distribution $F$. We recover Definition~\ref{def:optimal_bound} when 
\begin{align}
    \cC(\x) =\{ F \in \F:  F^{\ge}_{ T(\X)}(T(\x)) > \alpha \}, 
\end{align}
or equivalently the acceptance region of the test is defined as: 
\begin{align}
\cA(F):= \{ \x \in \R^n : F^{\ge}_{ T(\X)}(T(\x)) > \alpha \}.
\end{align}
 \begin{definition}
 \label{def:extended_optimal_bound}
Let $\F$ be a set of distributions and suppose that $\theta(F)$ exists for all $F \in \F$. Given a sample $\x \in \R^n$ and a set of samples  $\cA(F)$, define: 
$$
h^{\F}_{\cA}(\x)=\inf_F \; \theta(F)$$
subject to 
\begin{enumerate}
\item $F \in \cC(\x)$
\end{enumerate}
and let $h^{\F}_{\cA}(\x):=\infty$ if the problem is infeasible. 
\end{definition}
We will show in Theorem~\ref{thm:extended_correct} that if for any distribution $F \in \mathcal{F}$, $\P_F(  \X \in \cA(F)) \ge 1 - \alpha$ then  $h^{\F}_{\cA}(\X)$ is smaller than $\theta(F)$ with probability at least $1-\alpha$. 

\begin{theorem}[Validity]
\label{thm:extended_correct}
Let $\F$ be a set of distributions and suppose that $\theta(F)$ exists for all $F \in \F$.
For any $F \in \F$, if $\P_F( \X \in \cA(F)) \ge 1 - \alpha$ then the  bound produced by Definition~\ref{def:extended_optimal_bound} satisfies:
\begin{align*}
\P_F(h^{\F}_{\cA}(\X) \le \theta(F)) \ge 1-\alpha. 
\end{align*}
\end{theorem}
\begin{proof}
Let $\cC(\x):=\{F \in \F: \x \in \cA(F)\}.$
If $F \in \F$ and $\x \in \cA(F)$, then $F \in \cC(\x)$   and $h^{\F}_{\cA}(\x) \le \theta(F)$ by definition. Therefore for any $F \in F$: 
\begin{align*}
    \P_{F}(h^{\F}_{\cA}(\X) \le \theta(F)) 
     &\ge  \P_{F}( F \in \cC(\X)) \\
    &\ge  \P_{F}( \X \in \cA(F)) \\
    &\ge 1 - \alpha. 
\end{align*}
\end{proof}

We give some examples of $\mathcal{F}$ and $\cA(F)$ below: 
\begin{itemize}
  \item $\mathcal{F}$ is the set of all continuous distributions. Let $F^{\x}_n(y)$ denoted the empirical cdf created from sample $\x$: 
\begin{align}
    F^{\x}_n(y) = \frac{1}{n} \sum_{i=1}^n \mathbb{I}_{x_i \le y}, ~\forall y \in \R. 
\end{align}
$\cA(F)$ is then defined as:
\begin{align}
    \cA(F) = \{ \x \in \R^n :  \sup_{y \in \R} | F(y) - F^{\x}_n(y) | \le \sqrt{\frac{1}{2n} \log \frac{2}{\alpha}}\}. 
\end{align}
The by DKW inequality \citep{Massart1990}, for any $F \in \F$, we have: 
\begin{align}
    \P_F (\X \in \cA(F)) &= \P_F \left( \sup_{y \in \R} | F(y) - F^{\X}_n(y) | \le \sqrt{\frac{1}{2n} \log \frac{2}{\alpha}} \right) \\
    &\ge 1 - \alpha. 
\end{align}
\item Let $\mathcal{F}$ be the set of all distributions and $\cA(F) =\{ \x \in \R^n: F^{\ge}_{T(\X)}(T(\x)) > \alpha \}.$ Then: 
\begin{align}
        \P_F( \X \in \cA(F) ) &=  \P( F^{\ge}_{T(\X)}(T(\X)) > \alpha ) \\
        &= 1 - \P( F^{\ge}_{T(\X)}(T(\X)) \le \alpha ) \\
        &\ge 1 - \alpha \text{ by Lemma~\ref{lem:cdf_vs_U}}. 
    \end{align}
    \item $\mathcal{F}$ is the set of all distributions on $[0,1]$. We define: 
    \begin{align}
        \cA(F) =\left\{\x \in \R^n: \mu(F) >\bar{\x} - \sqrt{\frac{\ln(1/\alpha)}{2n}} \right\}, 
    \end{align}
     where $\bar{\x}_n$ denote the sample mean of $\x$.
    From Hoeffding's inequality for any $F \in \mathcal{F}$:
    \begin{align}
        \P_F( \X \in \cA(F) ) &=  \P_F\left( \mu(F) > \bar{\X} - \sqrt{\frac{\ln(1/\alpha)}{2n}} \right) \\ &\ge 1- \alpha. 
    \end{align}
\item Let $\mathcal{F}$ be the set of all distributions on $[0,1]$. We define
\begin{align}
&\cA_1(F) =\{ \x \in \R^n: F^{\ge}_{T(\X)}(T(\x)) > \alpha/2 \} \\
&\cA_2(F) =\left\{\x \in \R^n: \mu(F) >\bar{\x}_n - \sqrt{\frac{\ln(2/\alpha)}{2n}} \right\} \\
&\cA(F) = \cA_1(F) \cup \cA_2(F).
\end{align}
We have $\P( \X \notin \cA_1(F) ) \le \alpha/2$ and $\P( \X \notin \cA_2(F) ) \le \alpha/2$. Therefore:
\begin{align}
    &\P( \X \in \cA(F) ) =  \P( \X \in \cA_1(F) \text{ and } \X \in \cA_2(F) ) \\
    &= 1 - \P( \X \in \cA_1(F) \text{ or } \X \in \cA_2(F) ) \\
    &\ge 1 - \P( \X \in \cA_1(F)) - \P( \X \in \cA_2(F) ) \\
    &= 1 - \alpha/2 - \alpha/2 \\
    &= 1- \alpha.
\end{align}
  
\end{itemize}

\section{Related Works}
\label{app:related_works}

\citep{phan2021practical} is the only previous work we are aware of that can be used implicitly to automate the relaxation for many order functions $T$ for non-parametric confidence bounds, as proposed by our work. \cite{phan2021practical} do not explicitly present the result as automating the relaxation, but it can be used as such. The bound in \citep{phan2021practical} is valid for many order functions $T$ but does not have the guarantee to be arbitrarily close to the $T$-optimal bound as ours partially has (Section~\ref{sec:analysis}), and therefore can be strictly looser than the $T$-optimal bound.

We discuss previous works that try to find to solve Definition~\ref{def:optimal_bound} on a case-by-case basis. It can be done by finding an upper bound of: 
\begin{align}
\label{eq:tail}
    \P_G(T(\X) \ge T(\x)).
\end{align}

Below we discuss the results regarding the upper bound of Eq.~\ref{eq:tail} for some specific functions $T$ (Section~\ref{sec:specific_T}), and the results regarding finding a good function $T$ that will result in tight bounds for Definition~\ref{def:optimal_bound} (Section~\ref{sec:good_T}). 
\subsection{Results for Some Order Functions \texorpdfstring{$T$}{T}}
\label{sec:specific_T}
Let $\F$ be the set of all distributions on $[0,1]$ and $T$ be the sample mean. When $\X$ are i.i.d. samples, \citet{luczak2016maximaltailprobabilitysums} find a tight upper bound when $T(\x) \le 1/n$. When $\X$ are independent (and not necessarily identical) samples with the same mean, \citet{bentkus} finds an upper bound for any $\x$ such that $T(\x) > 0$. When $\X$ are i.i.d. samples, \citet{Bentkus01aninequality} provides a tighter bound than \citep{bentkus} for the special case when $\E[X_i] = 1/2$. 

Let $\F$ be the set of all distributions on $[0,1]$ and $T$ be a linear function. When $\X$ are independent (and not necessarily identical) samples with the same mean, \cite{pinelis2006inequalities} also provides an upper bound of the tail bound in Eq.~\ref{eq:tail} for the special case when $\mu(G) = 1/2$.     

Anderson's bound \citep{Anderson1969} is a linear function of the order statistics of the samples, so if we prove a generalization of the result in \citep{luczak2016maximaltailprobabilitysums} that involves a linear function of the order statistics of the samples instead of the sum, we can use Anderson's bound as function $T$ and have a bound for $\F^{[0,\infty)}$ that is better than or equal to Anderson for any sample.  \cite{phan2021practical} also provide an upper bound for any $T$, which leads to results proven to be better than or equal to Anderson's when the support is $[0,1]$ and equal to Anderson's when the support is $[0, \infty)$.  Anderson's bound   is shown to be tighter than Hoeffding's for any sample \citep{Learned-MillerThomas2019}.

\subsection{Evaluating the Performance of the Order Function \texorpdfstring{$T$}{T}}
\label{sec:good_T}

When the set of distributions $\F$ is a family of distribution with unknown parameter $\theta$ such that the family of probability density functions or probability mass functions $\{g(t|\theta): \theta \in \Theta\}$ has a monotone likelihood ratio and $T$ is a sufficient statistic for $\theta$, then Karlin-Rubin theorem states that the test that use $T$ is a uniformly most powerful (UMP) test \citep{CaseBerg:01}. From the Ghost-Pratt identity, using $T$ will result in the bound with the shortest quantity $\E_{\theta_0} [b^{\mathcal{F}, \alpha}_T(\x) - \theta_0]$ \citep{casella1996ghosh, madansky1962more} (which is called the expected excess by \citet{madansky1962more}).

\end{document}